\documentclass[format=acmsmall, review=false]{acmart}

\usepackage{acm-ec-20} 

\usepackage{amsthm}

\usepackage{algorithm}
\usepackage[noend]{algpseudocode}

\usepackage{pgfplots}

\usepackage{tcolorbox}

\setcitestyle{acmnumeric}

\newcommand*{\draft}{}

\newcommand{\algrule}{\Statex \hrulefill}

\ifdefined\draft
    \newcommand{\patrick}[1]{\textcolor{blue}{[Patrick: #1]}}
    \newcommand{\andy}[1]{\textcolor{red}{[Andy: #1]}}
    \newcommand{\brendan}[1]{\textcolor{orange}{[Brendan: #1]}}
\else
    \newcommand{\patrick}[1]{}
    \newcommand{\andy}[1]{}
    \newcommand{\brendan}[1]{}
\fi

\begin{document}

\title{Minimmit: Fast Finality with Even Faster Blocks}

\author{Brendan Kobayashi Chou\textsuperscript{1}}
\author{Andrew Lewis-Pye\textsuperscript{1,2}}
\author{Patrick O'Grady\textsuperscript{1}}

\renewcommand{\shortauthors}{Chou et al.}

\affiliation{%
\textsuperscript{1}\institution{Commonware}
\country{USA}
}

\affiliation{%
\textsuperscript{2}\institution{London School of Economics}
\country{UK}
}



\date{August 2025}

\begin{abstract}
Achieving low-latency consensus in geographically distributed systems remains a key challenge for blockchain and distributed database applications. To this end, there has been significant recent interest in State-Machine-Replication (SMR) protocols
that achieve \emph{2-round finality} under the assumption that $5f+1\leq n$, where $n$ is the number of processors and $f$ bounds the number of processors that may exhibit Byzantine faults. In these protocols, instructions are organised into \emph{views}, each led by a different designated leader, and 2-round finality means that a leader's proposal can be finalised after just a single round of voting, meaning two rounds overall (one round for the proposal and one for voting).

We introduce Minimmit, a Byzantine-fault-tolerant SMR protocol with lower latency than previous 2-round finality approaches.  Our key insight is that view progression and transaction finality can operate on different quorum thresholds without compromising consistency or liveness. Experiments simulating a globally distributed network of 50 processors, uniformly assigned across ten virtual regions, show that the approach leads to a 23.1\% reduction in view latency and a 10.7\% reduction in transaction latency compared to the state-of-the-art.
\end{abstract}

\maketitle

\section{Introduction} \label{intro}
Protocols for State-Machine-Replication (`blockchain' protocols) are commonly designed to tolerate periods of unreliable message delivery. Formally, this means that protocols are required to satisfy \emph{consistency} (all correct processors agree on the sequence of finalised transactions) and \emph{liveness} (transactions are eventually finalised by correct processors) in the \emph{partially synchronous} setting. In this setting, it is known \cite{DLS88} that satisfying both liveness and consistency is possible precisely if $n\geq 3f+1$, where $n$ is the number of processors and $f$ bounds the number of processors that may exhibit Byzantine faults.

 \vspace{0.1cm}
A key metric when considering the efficiency of State-Machine-Replication (SMR) protocols  is \emph{transaction latency}, i.e., the time it takes for transactions to be \emph{finalised}.
Most protocols organise operations into \emph{views}, led by designated \emph{leaders}, making the number of communication rounds per view a crucial consideration in analysing latency.
 Under the standard assumption that $n\geq 3f+1$, it is known \cite{abraham2021good} that `3-round' finality is optimal: here 3-round finality means that a leader's proposal can be finalised after two further rounds of voting (this gives three communication rounds overall: round 1 for the proposal and rounds 2 and 3 for voting). Accordingly, standard protocols such as PBFT \cite{castro1999practical} and Tendermint \cite{buchman2016tendermint,buchman2018latest} have 3-round finality. On the other hand, it has been known for some time \cite{martin2006fast} that 2-round finality (one round for the proposal and one further round of voting) is possible under the assumption that $n\geq 5f+1$. Recent research has seen a significant renewed interest in approaches to 2-round finality: In the last year, Matter Labs released ChonkyBFT \cite{francca2025chonkybft}, Offchain Labs released Kudzu \cite{shoup2025kudzu},  Anza Labs and Solana  released Alpenglow \cite{alpen}, and Supra Research and Espresso Systems released Hydrangea \cite{shrestha2025hydrangea}. In part, the motivation for this focus on the $n\geq 5f+1$ assumption is driven by the scale of modern blockchain systems. When SMR protocols were first deployed in the 1980s, $n$ was small. Today, a typical blockchain system may consist of thousands of processors, meaning that Byzantine attacks on a significant fraction  of participants would be extremely costly, and are deemed unlikely.  This makes it appropriate to shift the resilience requirement from $n\geq 3f+1$ to $n\geq 5f+1$, since doing so allows for a significant reduction in  latency.

 \vspace{0.1cm}
This paper describes Minimmit, which differs from previous protocols with 2-round finality by foregoing the \emph{slow path} that is used by most previous protocols when 2-round finality fails, and by allowing processors to proceed to the next view after receiving only $2f+1$ votes. The idea behind this tradeoff is that, in a globally distributed network of processors with a range of connection speeds, receiving a smaller quorum of votes will often take significantly less time than waiting for a larger quorum of votes. Requiring the smallest possible quorum for view progression therefore reduces view latency (time between views). This also reduces transaction latency, i.e., the time it takes to finalise transactions,  since transactions wait for less time before being included in a block.


 \subsection{Evidence for reduced latency}
 To evaluate the performance of Minimmit, we implemented a network simulator\footnote{\url{https://github.com/commonwarexyz/monorepo/tree/19f19d32760daf1d497295726ec92a1e6b84959f/examples/estimator}} that runs custom workloads over configurable network topologies. This simulator provides each participant the opportunity to drive the simulation (i.e. propose a block) while collecting telemetry by region and role. To model realistic network conditions, we apply latency and jitter to message broadcasts using empirical measurements of AWS inter-region performance over the public internet from the past year.

 \vspace{0.1cm}
 We distinguish \emph{view latency} (time between views), \emph{block latency} (time to finalise blocks), and \emph{transaction latency} (time to finalise transactions). To analyse how different protocols compare across these metrics, we tested Simplex \cite{chan2023simplex}, Kudzu \cite{shoup2025kudzu}, and Minimmit on a globally distributed network of 50 processors uniformly assigned across ten ``virtual'' regions: us-west-1, us-east-1, eu-west-1, ap-northeast-1, eu-north-1, ap-south-1, sa-east-1, eu-central-1, ap-northeast-2, and ap-southeast-2. While we implemented Simplex specifically, finalisation times will be similar for other standard 3-round finality protocols that employ `all-to-all' message sending, such as PBFT and Tendermint, since they use the same structure of leader proposal, followed by two rounds of voting. Likewise, Kudzu is representative of other 2-round finality protocols that employ `all-to-all' message sending and a fast path mechanism, such as Alpenglow\footnote{As discussed in Section \ref{exp}, direct comparisons with Alpenglow are complicated by the fact that Alpenglow incorporates a scheme for disseminating erasure-coded block data during fixed
400ms slots, and so is non-responsive. The fixed 400ms window causes significantly increased view latency for Alpenglow for the parameters considered in our experiments.} and Hydrangea.  Under these conditions, Simplex achieves view latency  194ms (with standard deviation $\sigma$=30ms) and block latency  299ms ($\sigma$=26ms). Kudzu achieves view latency 190ms and block latency  220ms ($\sigma$=29ms). Minimmit achieves view latency 146ms ($\sigma$=21ms)  and block latency 220ms ($\sigma$=28ms). Minimmit reduces view latency by 25\% compared to Simplex and 23\% compared to Kudzu. Minimmit reduces block latency by 26\% compared to Simplex, equivalent to Kudzu.

 \vspace{0.1cm}
 For fixed sized blocks, decreased view latency clearly translates into increased throughput. However, it also  produces  reduced transaction latency for users. Consider a transaction submitted immediately after block construction at `height'\footnote{The \emph{height} of a block is its number of ancestors. Ancestors are formally defined in Section \ref{setup}.} $h$. This transaction must await inclusion in a block of height $h+1$ and for that block's subsequent finalisation. In recently proposed protocols operating under the $n\geq 5f+1$ assumption -- such as Kudzu, Alpenglow, and Hydrangea -- view progression requires $n-2f$ votes. Consequently, transaction finalisation requires waiting for both the completion of the current view (190ms) and the finalisation of the subsequent block containing the transaction (220ms), totalling 410ms in this model. In contrast, Minimmit's view progression threshold of $2f+1$  enables the same transaction to achieve finalisation in 146ms + 220ms = 366ms, representing a 10.7\% reduction in end-to-end latency.

\subsection{Our contributions}
Our contributions are as follows:
\begin{enumerate}
\item We describe a novel view-change mechanism to reduce view latency;
\item We give formal proofs of safety, liveness and optimistic responsiveness,\footnote{Optimistic responsiveness will be defined in Section \ref{resp}.} under the $n\geq 5f+1$ assumption;
\item We carry out experimental evaluations, showing roughly a 23\% reduction in view latency and an 11\% reduction in transaction latency compared to the state-of-the-art.
\item In Section \ref{opt}, we describe (along with other optimisations) a mechanism that uses aggregate signatures to bound the communication required to achieve liveness after periods of asynchrony. This approach may also be applied to improve efficiency for other protocols in the `Simplex mould'.
\item In Appendices \ref{erasuresetup}--\ref{erasureanalysis} we show how to implement a version of Minimmit that uses erasure coding with a \emph{data expansion} rate of 2.5. If  blocks are of size $B$, this means that the total amount of block data sent by each leader (to all processors combined) is roughly $2.5B$.  
\end{enumerate}

\subsection{Dropping the slow path}
Many 2-round finality protocols implement a `slow path' to be used when 2-round finality fails: generally, this slow path just requires an extra round of voting.  Since Minimmit achieves improved latency by foregoing the `slow path', it is interesting to analyse any resulting sacrifice in resilience, when compared to other leading protocols.

\vspace{0.1cm}
\noindent \textbf{Alpenglow}. Alpenglow is formally analysed under the same assumption that $n\geq 5f+1$. While there is some consideration of circumstances in which the protocol can tolerate a further $f$ crash failures, the required assumptions for this case (essentially that Byzantine leaders cannot carry out a form of proposal equivocation) do not hold under partial synchrony.

\vspace{0.1cm}
\noindent \textbf{Kudzu}. Kudzu makes the more general assumption that $n\geq 3f+2p+1$ for a tunable parameter $p$. Of course, in the case that $p=f$, this corresponds to the same assumption that $n\geq 5f+1$. The protocol guarantees liveness and consistency so long as at most $f$ processors display Byzantine faults. It also guarantees that, during synchrony,  a correct leader can finalise a block after a single round of voting, so long as the number of faulty processors is $\leq p$. For a correct leader during synchrony, the `slow path' (requiring two rounds of voting) is therefore only relevant if the number of faulty processors is strictly between $p$ and $f$. In this case, the quorum required for finality after a single round of voting is also larger than for Minimmit ($n-p$ rather than $n-f$). Kudzu therefore gives no strict `like-for-like' improvement in resilience compared to Minimmit, although the tunable parameter that Kudzu provides is beneficial.

\vspace{0.1cm}
\noindent \textbf{Hydrangea}. Compared to Minimmit, Alpenglow, and Kudzu, Hydrangea has improved resilience to crash failures. As before, consider the case of a correct leader during synchrony. For a parameter $k\geq 0$, and for a system of $n = 3f + 2c + k + 1$ processors, Hydrangea achieves finality after 1 round of voting, so long as the number of faulty processors (Byzantine or crash) is at most $p = \lfloor \frac{c+k}{2} \rfloor$. In the case that $c=0$, this aligns precisely with the bounds provided by Kudzu. However, in more adversarial settings with up to $f$ Byzantine faults and $c$ crash faults,
Hydrangea also obtains finality after two rounds of voting. Of course, the benefit of Minimmit is that it achieves reduced transaction latency in the case that $n\geq 5f+1$.

\subsection{Paper structure}
The remainder of the paper is structured as follows:
\begin{itemize}
\item Section \ref{setup} describes the formal setup.
\item Section \ref{intu} describes the intuition behind Minimmit.
\item Section \ref{formal} formally specifies the Minimmit protocol.
\item Section \ref{analysis} gives formal proofs of consistency and liveness.
\item Section \ref{opt} describes a number of optimisations, such as the use of threshold signatures, erasure coding, and the use of aggregate signatures to speed up message dissemination and fast recovery after periods of asynchrony.
\item Section \ref{exp} describes our experiments and results.
\item Section \ref{rw} describes related work.
\item Section  \ref{fc} contains some final comments and conclusions.
\item Appendix  \ref{app:repro} describes how to reproduce our experiments using the \texttt{commonware-estimator}.\footnote{\url{https://github.com/commonwarexyz/monorepo/tree/19f19d32760daf1d497295726ec92a1e6b84959f/examples/estimator}}
\item Appendix   \ref{erasuresetup} describes the setup required for erasure coding. 
\item Appendix   \ref{erasureintu} gives the intuition behind our efficient integration of erasure codes. 
\item Appendix  \ref{erasurespec} gives a formal specification of  E-Minimmit (Minimmit with erasure codes). 
\item Appendix   \ref{erasureanalysis} establishes consistency and liveness for E-Minimmit. 
\end{itemize}

\section{The Setup} \label{setup}

We consider a set $\Pi = \{p_1, \ldots, p_{n} \}$ of $n$ processors. For $f$ such that $5f+1\leq n$, at most $f$ processors may become corrupted by the adversary during the course of the execution, and may then display  \emph{Byzantine} (arbitrary) behaviour.
Processors that never become corrupted by the adversary are referred to as \emph{correct}.

\vspace{0.2cm}
\noindent \textbf{Cryptographic assumptions}. Our cryptographic assumptions are standard for papers on this topic. Processors communicate by point-to-point authenticated channels. We use a cryptographic signature scheme, a public key infrastructure (PKI) to validate signatures, and a collision resistant hash function $H$.\footnote{In Section \ref{opt} and the appendices,  we will also consider optimisations of the protocol that use threshold signatures, aggregate signatures, and erasure codes.}  We assume a computationally bounded adversary. Following a common standard in distributed computing and for simplicity of presentation (to avoid the analysis of negligible error probabilities), we assume these cryptographic schemes are perfect, i.e., we restrict attention to executions in which the adversary is unable to break these cryptographic schemes.

  \vspace{0.2cm}
\noindent \textbf{The partial synchrony model}. As noted above, processors communicate using point-to-point authenticated channels. We consider the standard partial synchrony model, whereby the execution is divided into discrete timeslots $t\in \mathbb{N}_{\geq 0}$ and a message sent at time $t$ must arrive at time $t'>t$ with $t'\leq \max\{\text{GST},t\} + \Delta$. While $\Delta$ is known, the value of GST is unknown to the protocol. The adversary chooses GST and also message delivery times, subject to the constraints already defined.
Correct processors begin the protocol execution before GST and are not assumed to have synchronised clocks.  For simplicity, we do assume that
the clocks of correct processors all proceed in real time,
meaning that if
$t'>t$ then the local clock of correct $p$ at time $t'$ is $t'-t$ in
advance of its value at time $t$. Using standard arguments, our
  protocol and analysis can
be extended in a straightforward way to the case
in which there is a known upper bound on the difference
between the clock speeds of correct processors.

\vspace{0.2cm}
\noindent \textbf{Transactions}. Transactions are messages of a
distinguished form,  signed by the \emph{environment}. Each timeslot, each processor may receive some
finite set of transactions directly from the environment. We make the standard assumption that transactions are unique (repeat transactions can be produced using an increasing `ticker'  or timestamps \cite{castro1999practical}).

\vspace{0.2cm}
\noindent \textbf{State machine replication}.
If $\sigma$ and $\tau$
are sequences, we write $\sigma \preceq \tau$ to denote that
$\sigma$ is a prefix of $\tau$.  We say $\sigma$ and~$\tau$ are
\emph{compatible} if $\sigma\preceq \tau$ or $\tau\preceq \sigma$.
If two sequences are not compatible, they are \emph{incompatible}.
If~$\sigma$ is a sequence of transactions, we write
$\text{tr}\in \sigma$ to denote that the transaction $\text{tr}$
belongs to the sequence $\sigma$.
Each processor $p_i$ is required to maintain an append-only log, denoted $\text{log}_i$, which at any timeslot is a sequence of distinct transactions. We also write $\text{log}_i(t)$ to denote the value $\text{log}_i$ at the end of timeslot $t$. The log being append-only means that for $t'>t$,  $\text{log}_i(t) \preceq \text{log}_i(t')$.
We
require the following conditions to hold in every execution:

\vspace{0.1cm}
\noindent \emph{Consistency}.   If $p_i$ and $p_j$ are correct, then for any timeslots $t$ and $t'$, $\text{log}_i(t)$ and $\text{log}_j(t')$ are compatible.

\vspace{0.1cm}
\noindent \emph{Liveness}. If $p_i$ and $p_j$ are correct and if $p_i$ receives the transaction $\text{tr}$ then, for  some  $t$,  $\text{tr}\in \text{log}_j(t)$.

\vspace{0.2cm}
\noindent \textbf{Blocks, parents, ancestors, and finalisation}. We specify a protocol that produces \emph{blocks} of transactions. Among other values, each block $b$ specifies a value $b.\text{Tr}$, which is a sequence of transactions. There is a unique \emph{genesis} block $b_{\text{gen}}$, which is considered \emph{finalised} at the start of the protocol execution. Each block $b$ other than the genesis block has a unique \emph{parent}. The \emph{ancestors} of $b$ are $b$ and all ancestors of its parent (while the genesis block has only itself as ancestor), and each block has the genesis block as an ancestor.  Each block $b$ thus naturally specifies an extended sequence of transactions, denoted $b.\text{Tr}^*$, given by concatenating the values $b'.\text{Tr}$ for all ancestors $b'$ of $b$, removing any duplicate transactions. When a processor $p_i$ \emph{finalises} $b$ at timeslot $t$, this means that, upon obtaining all ancestors of $b$,  it sets $\text{log}_i(t)$ to extend $b.\text{Tr}^*$.
Two blocks are \emph{inconsistent} if neither is an ancestor of the other.

\section{The intuition} \label{intu}

In this section, we give an informal account of the intuition behind the protocol design.

\vspace{0.2cm}
\noindent \textbf{One round of voting}. We aim to specify a standard form of \emph{view-based} protocol, in which each view has a designated \emph{leader}. If a view has a correct leader and begins after GST, the leader should send a block $b$ to all other processors, who will then send a  signed \emph{vote} for $b$ to all others. Upon receipt of $n-f$  votes for $b$ (by distinct processors), our intention is that processors should then be able to immediately finalise $b$: this is called \emph{2-round finality} \cite{martin2006fast} (one round to send the block and one round for voting). $5f-1\leq n$ is necessary and sufficient\footnote{See \cite{kuznetsov2021revisiting} and \url{https://decentralizedthoughts.github.io/2021-03-03-2-round-bft-smr-with-n-equals-4-f-equals-1/}} for 2-round finality: we make the assumption that $5f+1\leq n$ for simplicity.  We'll call a set of $n-f$ votes for a block an \emph{L-notarisation} (where `L' stands for `large').

\vspace{0.2cm}
\noindent \textbf{When to enter the next view?} We specified above that a single L-notarisation should suffice for finality. However, as well as achieving 2-round finality, we also wish  processors to proceed to the next view immediately upon seeing a smaller set of votes, which we'll call an \emph{M-notarisation} (where the `M' stands for `mini'). As explained in Section \ref{intro}, the rationale behind this is that, in realistic scenarios,  receiving $n-f$ votes will generally take much longer than receiving a significantly smaller set of votes (of size $2f+1$, say). Proceeding to view $v+1$ immediately upon seeing an M-notarisation for $b$ in view $v$ therefore speeds up the process of block formation: processors can begin view $v+1$ earlier, and then finalise $b$ upon later receiving the larger notarisation.

How many votes should we require for an M-notarisation? Note that, so long as correct processors don't vote for more than one block in a view, the following will hold:

\begin{enumerate}
\item[(X1)] If $b$ for view $v$ receives an L-notarisation, then no block $b'\neq b$ for view $v$ receives $2f+1$ votes.
\end{enumerate}
To see this, suppose towards a contradiction that $b$ for view $v$ receives an L-notarisation and  $b'\neq b$ for view $v$ receives $2f+1$ votes. Let $P$ be the set of processors that contribute to the L-notarisation for $b$, and let $P'$ be the set of processors that vote for $b'$.  Then  $|P\cap P'| \geq (n-f)+(2f+1)-n = f+1$. So, $P\cap P'$ contains a correct processor, which contradicts the claim that correct processors don't vote for two blocks in one view.

So, if we set an M-notarisation to be a set of $2f+1$ votes for $b$, and allow processors to proceed to view $v+1$ upon seeing an M-notarisation for $b$ in view $v$, then any processor receiving an M-notarisation for $b$ knows that no block $b'\neq b$ for view $v$ can receive an L-notarisation. In particular, it is useful to see things from the perspective of the leader, $p_i$ say, of view $v+1$. If $p_i$ has seen an M-notarisation for $b$ in view $v$, and so long as $p_i$ resends this to other processors,\footnote{To reduce communication complexity, this could be done using a threshold signature scheme, but we defer such considerations to Section \ref{opt}.} $p_i$ can be sure that all processors have proof that no block other than $b$ could have been finalised in view $v$. So, $p_i$ can propose a block with $b$ as parent, and all correct processors will have evidence that it is safe to vote for $p_i$'s proposal.

\vspace{0.2cm}
\noindent \textbf{Nullifications}. Since it may be the case that no block for view $v$ receives an M-notarisation (e.g., if the leader is Byzantine),  processors must sometimes produce signed messages indicating that they wish to move to view $v+1$ because of a lack of progress in view $v$: we'll call these nullify$(v)$ messages.\footnote{Our approach here is somewhat similar to Simplex \cite{chan2023simplex}, but the reader need not have familiarity with that protocol to follow the discussion.} We do not yet specify precisely when processors send $\text{nullify}(v)$ messages (we will do so shortly). For now, we promise only that a statement analogous to (X1) will hold for $\text{nullify}(v)$ messages:
\begin{enumerate}
\item[(X2)] If some block $b$ for view $v$ receives an L-notarisation, then it is not the case that $2f+1$ processors send $\text{nullify}(v)$ messages.
\end{enumerate}
With the promise of (X2) in place, let's define a \emph{nullification} to be a set of $2f+1$ nullify$(v)$ messages (signed by distinct processors). We specify that a processor should enter view $v+1$ upon receiving either:
\begin{itemize}
\item  An M-notarisation for view $v$, or;
\item  A nullification for view $v$.
\end{itemize}
\noindent So long as (X2) holds, any processor receiving a nullification for view $v$ knows that no block for view $v$ can receive an L-notarisation.

\vspace{0.2cm}
\noindent \textbf{Defining valid proposals so as to maintain consistency}. Suppose $p_i$ is the leader of view $v$:
\begin{enumerate}
\item[(i)]  Upon entering view $v$, $p_i$ finds the greatest $v'<v$ such that it has received an M-notarisation for some block $b$ for view $v'$: since $p_i$ has entered view $v$, it must have received nullifications for all views in $(v',v)$. Processor $p_i$ then proposes a block $b'$ with $b$ as parent.
\item[(ii)]  Other processors will vote for $b'$, so long as they see nullifications for all views in $(v',v)$ and an M-notarisation for $b$: since all processors (including $p_i$) will resend notarisations and nullifications upon first receiving them, if view $v$ begins after GST, $p_i$ can therefore be sure that all correct processors will receive the messages they need in order to vote for $b'$.
\end{enumerate}

It should also not be difficult  to see that this will guarantee consistency (see Section \ref{analysis} for the full proof). Towards a contradiction, suppose that $b_1$ for view $v_1$ receives an L-notarisation, and that there is a least view $v_2\geq v_1$ such that some block $b_2$ for view $v_2$ that does not have  $b_1$ as an ancestor receives an M-notarisation. From (X1) it follows that $v_2>v_1$. By our choice of $v_2$, and since correct processors will not vote for blocks until they see an M-notarisation for the parent, it follows that the parent of $b_2$ must be for a view $<v_1$. This gives a contradiction, because correct processors would not vote for $b_2$ in view $v_2$ without seeing a nullification for view $v_1$. Such a nullification cannot exist, by $(X2)$.

\vspace{0.2cm}
\noindent \textbf{Ensuring liveness}. To ensure liveness, we must first guarantee that if all correct processors enter a view, then they all eventually leave the view. To this end we allow that, while correct processors will not vote in any view $v$ after sending a nullify$(v)$ message, they \emph{may} send nullify$(v)$ messages after voting. More precisely, correct $p_i$ will send a  $\text{nullify}(v)$ message while in view $v$ if either:
\begin{enumerate}
\item[(a)] Their `timer' for view $v$ expires (time $2\Delta$ passes after entering the view) \emph{before voting}, or;
\item[(b)] They receive messages from $2f+1$ processors, each of which is either:
\begin{itemize}
\item  A nullify$(v)$ message, or;
\item  A vote for a view $v$ block different than a view $v$ block that $p_i$ has already voted for.
\end{itemize}
\end{enumerate}
Combined with the fact that correct processors will forward on nullifications and notarisations upon receiving them, the conditions above achieve two things. First, they suffice to ensure that (X2) is satisfied. If $b$ for view $v$ receives an L-notarisation, then let $P$ be the correct processors that vote for $b$, let $P'=\Pi\setminus P$, and note that $|P'|\leq 2f$. No processor in $P$  can send a nullify$(v)$ message via (a) or vote for a view $v$ block other than $b$. It follows that no processor in $P$ can be caused to send a nullify$(v)$ message via (b). So, at most $2f$ processors
 can send nullify$(v)$ messages.

 \vspace{0.2cm}
The conditions above also suffice to ensure that if all correct processors enter a view $v$, then they all eventually leave the view. Towards a contradiction, suppose all correct processors enter view $v$, but it is not the case that they all leave the view. Since correct processors forward nullifications and notarisations upon receiving them, this means that no correct processor leaves view $v$. Each correct processor eventually receives, from at least $n-f$ processors, either a vote for some block for view $v$ or a nullify$(v)$ message. If any correct processor receives an M-notarisation for the block they voted for, then we reach an immediate contradiction. So, suppose otherwise. If $p_i$ is a  correct processor that votes for a view $v$ block, it follows that $p_i$  receives messages from at least $(n-f)-(2f)=n-3f\geq 2f+1$ processors, each of which is either a nullify$(v)$ message or a vote for a view $v$ block different than the view $v$ block that $p_i$ votes for. So, $p_i$ sends a nullify$(v)$ message via (b). Any correct processor that does not vote for a view $v$ block also sends a nullify$(v)$ message, so all correct processors send nullify$(v)$ messages, giving the required contradiction.

\vspace{0.2cm}
\noindent \emph{Adding an extra instruction to send votes}. Once we have ensured that correct processors progress through the views, establishing liveness amounts to showing that each correct leader after GST finalises a new block. This now follows quite easily, using the reasoning outlined in (i) and (ii) above, where we explained that (after GST)  processors are guaranteed to receive all the messages they need to verify the validity of a block proposed by a correct leader. However, a subtle issue does require us to stipulate one further context in which correct processors should vote for a block. Suppose view $v$ begins after GST and that the leader for view $v$ is correct. Since a correct processor $p_i$  proceeds to view $v+1$ immediately upon seeing an M-notarisation for a view $v$ block $b$, and since it is possible that this is received \emph{before} $p_i$ receives all nullifications required to verify that $p_i$ should vote for $b$, the possibility apparently remains that correct processors will proceed to view $v+1$ without $b$ receiving an L-notarisation, i.e., we do not yet have any guarantee that all correct processors will vote for $b$.  To avoid this, we stipulate that, if $p_i$ receives an M-notarisation for $b$ and has not previously sent a nullify$(v)$ message or voted during view $v$, then $p_i$ should vote for $b$ before entering view $v+1$. In this case, the fact that $b$ has already received an M-notarisation means that some correct processors have already voted for $b$, so it is safe for $p_i$ to do the same.
The formal proof appears in Section \ref{analysis}.

\vspace{0.4cm}
\noindent \textbf{Intuition summary (informal)}:
\begin{itemize}
\item  In each view, the leader proposes a block and all processors then vote (one round of voting). Correct processors vote for at most one block in each view, which ensures (X1).
\item An L-notarisation suffices for finalization, while an M-notarisation or a nullification suffices to move to the next view.
\item  Processors only vote for the block $b$ for view $v$, with parent $b'$ for view $v'$,  if they have seen a notarisation for $b'$ and nullifications for all views in $(v',v)$.  Using (X1) and (X2), this suffices for safety.
\item To ensure progression through the views, processors send nullify$(v)$ messages upon timing out ($2\Delta$ after entering the view) before voting, or upon receiving proof that no block for view $v$ will receive an L-notarisation. This ensures (X2) is satisfied.
\item Once we are given that processors progress through views, liveness follows from the fact that each correct  leader after GST will finalise a new block: correct processors will vote for the leader's proposal because the leader themself will have sent all messages required to verify the validity of the block. Since correct processors also vote for the leader's block in view $v$ upon seeing an M-notarisation for it (if they have not already voted in view $v$, or sent a nullify$(v)$ message), this ensures all correct processors vote for the block before leaving view $v$, and it receives an L-notarisation.

\end{itemize}

\section{The formal specification} \label{formal}
We initially give a specification aimed at simplicity. In Section \ref{opt}, we will describe optimisations, such as the use of threshold signatures to reduce communication complexity.
In what follows, we suppose that all messages are signed by the sender. We say `disseminate' to mean `send to all processors'. When a correct processor is instructed to send a message to itself, it regards that message as immediately
received.  For the sake of simplicity, we also initially assume (without explicit mention in the pseudocode) that correct processors automatically send new transactions to all others upon first receiving them - we will revisit this assumption in Section \ref{opt}. The pseudocode uses a number of message types, local
variables, functions and procedures, detailed below.

\vspace{0.2cm}
\noindent \textbf{The function} $\mathtt{lead}(v)$. The value $\mathtt{lead}(v)$ specifies the leader for view $v$. To be concrete, we set\footnote{Of course, leaders can also be randomly selected, if given an appropriate source of common randomness.} $\mathtt{lead}(v):= p_{j+1}$, where $j= v \text{ mod }n$.

\vspace{0.2cm}
\noindent \textbf{Blocks}. The \emph{genesis block} is the tuple $b_{\text{gen}}:=(0,\lambda,\lambda)$, where $\lambda$ denotes the empty sequence (of length 0).  A block other than the genesis block is a tuple $b=(v,\text{Tr}, h)$, signed by $\mathtt{lead}(v)$, where:
\begin{itemize}
\item $v\in \mathbb{N}_{\geq 1}$ (thought of as the view corresponding to $b$);
\item $\text{Tr}$ is a sequence of distinct transactions;
\item $h$ is a hash value (used to specify $b$'s parent).
\end{itemize}
We also write $b.\text{view}$, $b.\text{Tr}$ and $b.\text{par}$ to denote the corresponding entries of $b$. If $b.\text{view}=v$, we also refer to $b$ as a `view $v$ block'.

\vspace{0.2cm}
\noindent \textbf{Votes}. A \emph{vote} for the block $b$ is a message of the form $(\text{vote},b)$.

\vspace{0.2cm}
\noindent \textbf{M-notarisations}. An M-\emph{notarisation} for the block $b$ is a set of $2f+1$ votes for $b$, each signed by a different processor. (By an M-\emph{notarisation}, we mean an M-notarisation for some block.)

\vspace{0.2cm}
\noindent \textbf{L-notarisations}. An \emph{L-notarisation} for the block $b$ is a set of $n-f$ votes for $b$, each signed by a different processor. We note that, since an L-notarisation is a larger set of votes than an M-notarisation, if $p_i$ has received an L-notarisation for $b$, then it has necessarily received an M-notarisation for $b$.

\vspace{0.2cm}
\noindent \textbf{Nullify($v$) messages}. For $v\in \mathbb{N}_{\geq 1}$, a nullify$(v)$ message is of the form $(\text{nullify},v)$.

\vspace{0.2cm}
\noindent \textbf{Nullifications}. A \emph{nullification} for view $v$ is a set of $2f+1$ nullify$(v)$ messages, each signed by a different processor. (By a \emph{nullification}, we mean a nullification for some view.)

\vspace{0.2cm}
\noindent \textbf{The local variable} $\mathtt{S}$. This variable is maintained locally by each processor $p_i$ and stores all messages received. It is considered to be automatically updated, i.e., we do not give explicit instructions in the pseudocode updating $\mathtt{S}$. We say a set of messages $M'$ is \emph{contained in} $\mathtt{S}$ if $M'\subseteq \mathtt{S}$. We also regard $\mathtt{S}$ as containing a block $b$ whenever $\mathtt{S}$ contains any message (tuple) with $b$ as one of its entries. Initially, $\mathtt{S}$ contains only $b_{\text{gen}}$  and an L-notarisation (and an M-notarisation) for $b_{\text{gen}}$.

\vspace{0.2cm}
\noindent \textbf{The local variable} $\mathtt{v}$. Initially set to 1, this variable specifies the present view of a processor.

\vspace{0.2cm}
\noindent \textbf{The local timer} $\mathtt{T}$. Each processor $p_i$ maintains a local timer $\mathtt{T}$, which is initially set to 0 and increments in real-time. (Processors will be explicitly instructed to reset their timer to 0 upon entering a new view.)

\vspace{0.2cm}
\noindent \textbf{The local variables} $\mathtt{nullified}$, $\mathtt{proposed}$, and $\mathtt{notarised}$. These are used by $p_i$ to record whether it has yet sent a nullify$(\mathtt{v})$ message, whether it has yet proposed a block for view $v$, and the block  it has voted for in the present view: $\mathtt{nullified}$ and $\mathtt{proposed}$ are initially set to false, while $\mathtt{notarised}$ is initially set to $\bot$ (a default value different than any block). These values will be explicitly reset upon entering a new view.

\vspace{0.2cm}
\noindent \textbf{The function} SelectParent$(\mathtt{S},\mathtt{v})$. This function is used by the leader of a view, $p_i$ say,  to select the parent block to build on. If $v'<\mathtt{v}$ is the greatest view\footnote{There must exist such a view, since $\mathtt{S}$ always contains an M-notarisation for the genesis block.} such that $\mathtt{S}$ contains an M-notarisation for some $b$ with $b.\text{view}=v'$, and if $b$ is the lexicographically least such block, the function outputs $b$.

\vspace{0.2cm}
\noindent \textbf{The procedure} ProposeChild$(b,v)$. This procedure is executed by the leader $p_i$ of view $v$ to determine a new block. To execute the procedure, $p_i$:
\begin{itemize}
\item  Forms a sequence of distinct transactions Tr, containing all transactions received by $p_i$ and not included in $b'.\text{Tr}$ for any $b'\in \mathtt{S}$ which is an ancestor of $b$, and;
\item Disseminates the block $(v,\text{Tr},H(b))$.
\end{itemize}

 \vspace{0.2cm}
\noindent \textbf{When} $\mathtt{S}$ \textbf{contains a valid proposal for view} $v$. This condition is satisfied when $\mathtt{S}$ contains:
\begin{enumerate}
\item[(i)] Precisely one block $b$ of the form  $b=(v,\text{Tr},h)$ signed by $\mathtt{lead}(v)$;
\item[(ii)] An M-notarisation for some $b'$ with $H(b')=h$, and with $b'.\text{view}=v'$ (say), and;
\item[(iii)] A nullification for each view in the open interval $(v',v)$.
\end{enumerate}
When (i)--(iii) are satisfied w.r.t.\ $b$, we say $\mathtt{S}$ contains a valid proposal $b$ for view $v$.

\vspace{0.2cm}
\noindent \textbf{New nullifications and notarisations}. Processors will be required to forward on all newly received nullifications and M-notarisations. To make this precise, we must specify what is to count as a `new' nullification/notarisation. At timeslot $t$, $p_i$ regards a nullification $N\subseteq \mathtt{S}$ for some view $v$ (not necessarily equal to $\mathtt{v}$) as \emph{new} if:
\begin{itemize}
\item $\mathtt{S}$ (as locally defined) did not contain a nullification for view $v$ at any smaller timeslot, and;
\item $N$ is lexicographically least amongst nullifications  for view $v$ contained in $\mathtt{S}$.
\end{itemize}
At timeslot $t$, $p_i$ regards an $M$-notarisation $Q\subseteq \mathtt{S}$ for block $b$ as \emph{new}\footnote{Since it is not necessary for liveness, our pseudocode does not require processors to forward L-notarisations. One could also require processors to forward L-notarisations, and the proofs of Section \ref{analysis} would go through unchanged.} if:
\begin{itemize}
\item $\mathtt{S}$ did not contain an M-notarisation for $b$ at any smaller timeslot, and;
\item $Q$ is lexicographically least amongst M-notarisations for $b$ contained in $\mathtt{S}$.
\end{itemize}

\vspace{0.2cm} For ease of reference, message types and local variables are displayed in Tables 1 and 2. The pseudocode appears in Algorithm 1.

\begin{table}[h!]
  \begin{center}
    \label{tab:table1}
    \begin{tabular}{l|l} 
      \textbf{Message type} & \textbf{Description} \\
      \hline
      $b_{\text{gen}}$ & The tuple $(0,\lambda,\lambda)$, where $\lambda$ is the empty string \\
      block $b\neq b_{\text{gen}}$  & A tuple $(v,\text{Tr},h)$, signed by $\mathtt{lead}(v)$ \\
      vote for $b$  & A message $(\text{vote},b)$\\
      nullify$(v)$ & A message of the form $(\text{nullify},v)$ \\
      nullification for $v$ & A set of $2f+1$ nullify$(v)$ messages, each signed by a different processor \\
      M-notarisation for $b$ & A set of $2f+1$ votes for $b$,  each signed by a different processor \\
      L-notarisation for $b$& A set of $n-f$ votes for $b$,  each signed by a different processor \\

    \end{tabular}
        \caption{Messages}

  \end{center}
\end{table}

\begin{table}[h!]
  \begin{center}
    \label{tab:table2}
    \begin{tabular}{l|l} 
      \textbf{Variable} & \textbf{Description} \\
      \hline
      $\mathtt{v}$ & Initially 1, specifies the present view \\
      $\mathtt{T}$ & Initially 0, a local timer reset upon entering each view \\
$\mathtt{nullified}$ & Initially false, specifies whether already sent nullify$(\mathtt{v})$ message \\
$\mathtt{proposed}$ & Initially false, specifies whether already proposed a block for view $v$ \\
$\mathtt{notarised}$ & Initially set to $\bot$, records block voted for in present view \\
$\mathtt{S}$ & Records all received messages, automatically updated \\
& Initially contains only $b_{\text{gen}}$ and M/L-notarisations for $b_{\text{gen}}$ \\

    \end{tabular}
        \caption{Local variables}

  \end{center}
\end{table}

 \begin{algorithm} \label{alg1}
\caption{: the instructions for $p_i$}
\begin{algorithmic}[1]

\State At every timeslot $t$:

   \State  \hspace{0.3cm} Disseminate new nullifications in $\mathtt{S}$; \label{forwardN} \Comment `new' as defined in Section \ref{formal}

      \State  \hspace{0.3cm} Disseminate new M-notarisations in $\mathtt{S}$; \label{forwardML}


 \State

 \State \hspace{0.3cm} If $p_i=\mathtt{lead}(\mathtt{v})$ and $\mathtt{proposed}=$ false:

%
     \State \hspace{0.6cm} $\text{ProposeChild}(\text{SelectParent}(\mathtt{S},\mathtt{v}),\mathtt{v})$; \label{sendblock}   \Comment Send out a new block

     \State \hspace{0.6cm} Set $\mathtt{proposed}:=$ true; 

     \State

      \State \hspace{0.3cm} If $\mathtt{S}$ contains a valid proposal $b$ for view $\mathtt{v}$:  \Comment As defined in Section \ref{formal}
      \State \hspace{0.6cm}  If $\mathtt{notarised}=\bot$ and $\mathtt{nullified}=$ false:  \label{votecheck}
      \State \hspace{0.9cm} Set $\mathtt{notarised}:=b$ and disseminate $(\text{vote},b)$; \label{vote1} \Comment Send vote

      \State
        \State \hspace{0.3cm} If $\mathtt{T}=2\Delta$, $\mathtt{nullified}=$ false and $\mathtt{notarised}=\bot$:  \label{timeout1}
        \State \hspace{0.6cm} Set $\mathtt{nullified}:=$true and disseminate $(\text{nullify},\mathtt{v})$; \label{timeout2} \Comment Send nullify$(\mathtt{v})$

        \State

        \State  \hspace{0.3cm} If  $\mathtt{S}$ contains a nullification for $\mathtt{v}$:
        \State  \hspace{0.6cm}  Set $\mathtt{v}:=\mathtt{v}+1$, $\mathtt{nullified}:=$ false, $\mathtt{proposed}:=$ false, $\mathtt{notarised}:=\bot$, $\mathtt{T}:=0$; \label{newview1} 
        \State \Comment Go to next view
         \State  \hspace{0.3cm} If $\mathtt{S}$ contains an  M-notarisation for some $b$ with $b.\text{view}=\mathtt{v}$:
          \State  \hspace{0.6cm} If  $\mathtt{notarised} =\bot$ and $\mathtt{nullified}=$ false, disseminate $(\text{vote},b)$;  \label{vote2} \Comment Send vote
         \State  \hspace{0.6cm}  Set $\mathtt{v}:=\mathtt{v}+1$, $\mathtt{nullified}:=$ false, $\mathtt{proposed}:=$ false, $\mathtt{notarised}:=\bot$, $\mathtt{T}:=0$; \label{newview2}
         \State  \Comment Go to next view

      \State

      \State  \hspace{0.3cm} If $\mathtt{nullified}=$ false, $\mathtt{notarised}\neq \bot$ and $\mathtt{S}$ contains $\geq 2f+1$ messages, each signed by a \label{beginN}
      \State  \hspace{0.3cm}  different processor, and each either:
      \State  \hspace{0.4cm} (i) A message $(\text{nullify},\mathtt{v})$, or;
      \State  \hspace{0.4cm} (ii) Of the form $(\text{vote},b)$ for some $b$ s.t.\ $b.\text{view}=\mathtt{v}$ and $\mathtt{notarised}\neq b$:  \label{endN}

         \State  \hspace{0.7cm} Set $\mathtt{nullified}:=$ true and disseminate $(\text{nullify},\mathtt{v})$; \label{sendN}

          \State \Comment Send nullify$(\mathtt{v})$ message upon proof of no progress for $\mathtt{v}$

          \State

          \State \hspace{0.3cm} If $\mathtt{S}$ contains a new L-notarisation for any block $b$:
           \State \hspace{0.6cm} Finalise $b$; \Comment Finalisation as specified in Section \ref{setup}

\end{algorithmic}
\end{algorithm}

\section{Protocol Analysis} \label{analysis}

\subsection{Consistency}  \label{consec}
We say block $b$ \emph{receives} an M-notarisation if $b=b_{\text{gen}}$ or at least $2f+1$ processors send votes for $b$. Similarly,  we say  $b$ receives an L-notarisation if $b=b_{\text{gen}}$ or at least $n-f$ processors send votes for $b$. View $v$ receives a nullification if at least $2f+1$ processors send nullify$(v)$ messages.

\begin{lemma}[One vote per view] \label{singlevote}
Correct processors vote for at most one block in each view, i.e., if $p_i$ is correct then,  for each $v\in \mathbb{N}_{\geq 1}$, there exists at most one $b$ with $b.\text{view}=v$ such that $p_i$ sends a message $(\text{vote},b)$.
\end{lemma}
\begin{proof}
Recall that $\bot$ is a default value, different than any block. Each correct processor's local value $\mathtt{notarised}$ is initially set to $\bot$, and is also set to $\bot$ upon entering  any view (lines \ref{newview1} and \ref{newview2}). A correct processor $p_i$  will only vote for a block if $\mathtt{notarised}=\bot$ (lines \ref{votecheck} and \ref{vote2}). The claim of the lemma holds because, upon voting for any block $b$, $p_i$  either sets $\mathtt{notarised}:=b$ (line \ref{vote1}) and then does not redefine this value until entering the next view, or else immediately enters the next view (lines \ref{vote2} and \ref{newview2}).
\end{proof}

\begin{lemma}[(X1) is satisfied] \label{X1}  If  $b$ receives an L-notarisation, then no block $b'\neq b$ with $b'.\text{view}=b.\text{view}$ receives an M-notarisation.
\end{lemma}
\begin{proof}
Given Lemma \ref{singlevote}, this now follows as in  Section \ref{intu}. Towards a contradiction, suppose that $b$  receives an L-notarisation and that  $b'\neq b$ with $b'.\text{view}=b.\text{view}$  receives an M-notarisation. Let $P$ be the set of processors that contribute to the L-notarisation for $b$, and let $P'$ be the set of processors that vote for $b'$.  Then  $|P\cap P'| \geq (n-f)+(2f+1)-n = f+1$. So, $P\cap P'$ contains a correct processor, which contradicts Lemma \ref{singlevote}.
\end{proof}

\begin{lemma}[(X2) is satisfied]  \label{X2} If  $b$ receives an L-notarisation and $v=b.\text{view}$, then view $v$ does not receive a nullification.
\end{lemma}
\begin{proof}
Towards a contradiction, suppose  $b$ receives an L-notarisation, $v=b.\text{view}$,  and view $v$ receives a nullification. Let $P$ be the correct processors that vote for $b$, let $P'=\Pi\setminus P$, and note that $|P'|\leq 2f$. Since view $v$ receives a nullification, it follows that some processor in $P$ must send a nullify$(v)$ message. So, let $t$ be the first timeslot at which some processor $p_i\in P$ sends such a message. Since $p_i$ cannot send a nullify$(v)$ message upon timeout (lines \ref{timeout1}-\ref{timeout2}), $p_i$ must send the nullify$(v)$ message at $t$ because the conditions of lines  \ref{beginN}-\ref{endN} hold for $p_i$ at $t$, i.e., $p_i$ must have received $\geq 2f+1$ messages, each signed by a   different processor, and each of the form:
\begin{enumerate}
\item[(i)]  $(\text{nullify},v)$, or;
\item[(ii)] $(\text{vote},b')$ for some $b'\neq b$ with  $b'.\text{view}=v$.
\end{enumerate}
By Lemma \ref{singlevote}, no processor in $P$ sends a message of form (ii). By our choice of $t$, no processor in $P$ sends  a message of form (i) prior to $t$. Combined with the fact that $|P'|\leq 2f$, this gives the required contradiction.
\end{proof}

\begin{lemma}[Consistency] The protocol satisfies Consistency.
\end{lemma}
\begin{proof}
Towards a contradiction, suppose that two inconsistent blocks, $b$ and $b'$ say, both receive L-notarisations. Without loss of generality suppose $b.\text{view}\leq b'.\text{view}$. Set $b_1:=b$ and $v_1:=b_1.\text{view}$. Then there is a \emph{least}  $v_2\geq v_1$ such that some block $b_2$ satisfies:
\begin{enumerate}
\item $b_2.\text{view}=v_2$;
\item   $b_1$ is not an ancestor of $b_2$, and;
\item $b_2$  receives an M-notarisation.
\end{enumerate}
 From Lemma \ref{X1},  it follows that $v_2>v_1$. According to clause  (ii) from the definition of when $\mathtt{S}$ contains a valid proposal for view $v_2$, correct processors will not vote for $b_2$ in line \ref{vote1} until they receive an M-notarisation for its parent, $b_0$ say.  Correct processors will not vote for $b_2$ in line \ref{vote2} until $b_2$ has already received an M-notarisation, meaning that at least $f+1$ correct processors must first vote for $b_2$ via line \ref{vote1}, and $b_0$ must receive an M-notarisation. By our choice of $v_2$, it follows that  $b_0.\text{view}<v_1$. This gives a contradiction, because, by clause (iii) from the definition of a valid proposal for view $v_2$,  correct processors would not vote for $b_2$  in line \ref{vote1} without receiving a nullification for view $v_1$.  By Lemma \ref{X2}, such a nullification cannot exist. So,  block $b_2$  cannot receive an M-notarisation (and no correct processor votes for $b_2$ via either line \ref{vote1} or \ref{vote2}).
\end{proof}

\subsection{Liveness} \label{livesec}

\begin{lemma}[Progression through views] Every correct processor enters every view $v\in \mathbb{N}_{\geq 1}$.
\end{lemma}
\begin{proof}  Towards a contradiction, suppose that some correct processor $p_i$ enters view $v$, but never enters view $v+1$.  Note that correct processors only leave any view $v'$ upon receiving either a nullification for the view, or else an M-notarisation for some view $v'$ block. Since correct processors forward new nullifications and notarisations upon receiving them (lines \ref{forwardN} and \ref{forwardML}), the fact that $p_i$ enters view $v$ but does not leave it  means that:
\begin{itemize}
\item All correct processors enter view $v$;
\item No correct processor leaves view $v$.
\end{itemize}

 Each correct processor eventually receives, from at least $n-f$ processors, either a vote for some view $v$ block, or a nullify$(v)$ message. If any correct processor receives an M-notarisation for a view $v$ block, then we reach an immediate contradiction. So, suppose otherwise. If $p_j$ is a  correct processor that votes for a view $v$ block $b$, it follows that $p_j$  receives messages from at least $(n-f)-(2f)=n-3f\geq 2f+1$ processors, each of which is either:
\begin{enumerate}
\item[(i)]  A nullify$(v)$ message, or;
\item[(ii)]  A vote for a view $v$ block different than $b$.
\end{enumerate}
 So, the conditions of lines \ref{beginN}-\ref{endN} are eventually satisfied, meaning that $p_j$ sends a nullify$(v)$ message (line \ref{sendN}). Any correct processor that does not vote for a view $v$ block also sends a nullify$(v)$ message, so all correct processors send nullify$(v)$ messages. All correct processors therefore receive a nullification for view $v$ and leave the view (line \ref{newview1}), giving the required contradiction.
\end{proof}

\begin{lemma}[Correct leaders finalise blocks] \label{L1} If $p_i=\mathtt{lead}(v)$ is correct, and if the first correct processor to enter view $v$ does so after GST, then $p_i$ disseminates a block  and that block  receives an L-notarisation.
\end{lemma}
\begin{proof} Suppose $p_i=\mathtt{lead}(v)$ is correct and that the first correct processor $p_j$ to enter view $v$ does so at timeslot $t\geq\text{GST}$. If $v>1$, processor $p_j$ enters view $v$ upon receiving either a nullification for view $v-1$, or else an M-notarisation for some view $v-1$ block. Since $p_j$ forwards on all new notarisations and nullifications that it receives (lines \ref{forwardN} and \ref{forwardML}), it follows that all correct processors enter view $v$ by $t+\Delta$ (note that this also holds if  $v=1$). Processor $p_i$ therefore disseminates a new block $b$  by $t+\Delta$, which is received by all processors by $t+2\Delta$. Let $b'$ be the parent of $b$ and suppose $b'.\text{view}=v'$.  Then $p_i$  receives an M-notarisation for $b'$ by $t+\Delta$. Since $p_i$ forwards on all new notarisations that it receives (line \ref{forwardML}), all correct processors receive an M-notarisation for $b'$ by $t+2\Delta$. Since $p_i$ has entered view $v$, it must also have received nullifications for all views in the open interval $(v',v)$ by $t+\Delta$, and all correct processors  receive these by $t+2\Delta$. All correct processors therefore vote for $b$ (by either line \ref{vote1} or \ref{vote2})\footnote{The point of line \ref{vote2} is to ensure this part of the proof goes through. As noted in Section \ref{intu}, without it, there is the possibility that correct processors move to the next view upon seeing an M-notarisation, before they are able to vote via line \ref{vote1}, thereby failing to guarantee an L-notarisation.} before any correct processor sends a nullify$(v)$ message. The block $b$ therefore receives an L-notarisation, as claimed.
\end{proof}

\begin{lemma}[Liveness] \label{L2}The protocol satisfies Liveness.
\end{lemma}
\begin{proof}
Suppose correct $p_i$ receives the transaction $\text{tr}$. Let $v$ be a view with $\mathtt{lead}(v)=p_i$ and such that the first correct processor to enter view $v$ does so after GST.
By Lemma \ref{L1}, $p_i$ will send a block $b$ to all processors, and $b$ will receive an L-notarisation. From the definition of the ProposeChild procedure, it follows that $\text{tr}$ will be included in $b'.\text{Tr}$ for some ancestor $b'$ of $b$, and all correct processors will add tr to their log upon receiving all ancestors of $b$ (see the final paragraph of Section \ref{setup}). Correct processors only vote for blocks whose parent has already received an M-notarisation. All ancestors of $b'$ of $b$ must therefore receive M-notarisations, meaning that at least $f+1$ correct processors disseminate each such $b'$, and correct processors receive all ancestors of $b$.
\end{proof}

\subsection{Optimistic responsiveness} \label{resp}
Let $\delta\leq \Delta$ be the \emph{actual} (unknown) least upper bound on message delay after GST, and let $f_a\leq f$ be the actual (unknown) number of Byzantine processors.  If a transaction $\text{tr}$ is first received by a correct processor at time $t$, and is first finalised by all correct processors (i.e., appended to $\text{log}_i$ for every correct $p_i$) at time $t+\ell$, then we say \emph{latency for }$\text{tr}$ is $\ell$. We say a protocol is \emph{optimistically responsive} if latency is $O(f_a\Delta+\delta)$ for all transactions that are first received by any correct processor after GST: in particular, this means latency after GST is $O(\delta)$ when processors act correctly. In this section we show that Minimmit is optimistically responsive. In fact, it satisfies the stronger condition that latency after GST is $O(\delta)$ when \emph{leaders} act correctly.\footnote{Since the notion of leaders is protocol specific, we prefer to use the more general definition as stated, but the stronger result also follows directly from the proofs given in this section.}

\begin{lemma} \label{olem1} Suppose $\mathtt{lead}(v)$ is correct and that the first correct processor to enter view $v$ does so at $t\geq $GST. Then all correct processors leave view $v$ and finalise a view $v$ block by $t+O(\delta)$.
\end{lemma}
\begin{proof} Proving the claim of the lemma just involves reviewing the proof of Lemma \ref{L1} and observing that the leader's block will actually be finalised by all correct processors within time $O(\delta)$ of any correct processor entering the view.

As before, suppose first $p_i=\mathtt{lead}(v)$ is correct and that the first correct processor to enter view $v$ does so at timeslot $t\geq\text{GST}$. Since correct processors forward  on all new notarisations and nullifications that they receive, it follows that all correct processors enter view $v$ by $t+\delta$.  Processor $p_i$ therefore disseminates a new block $b$  by $t+\delta$, which is received by all processors by $t+2\delta$. Let $b'$ be the parent of $b$ and suppose $b'.\text{view}=v'$. Then, from the definition of the function SelectParent, it follows that $p_i$  receives an M-notarisation for $b'$ by $t+\delta$. Since $p_i$ forwards on all new notarisations that it receives (line \ref{forwardML}), all correct processors receive an M-notarisation for $b'$ by $t+2\delta$. Since $p_i$ has entered view $v$, it must also have received nullifications for all views in the open interval $(v',v)$ by $t+\delta$, and all correct processors  receive these by $t+2\delta$. All correct processors therefore vote for $b$ (by either line \ref{vote1} or \ref{vote2}) by $t+2\delta$, and before any correct processor sends a nullify$(v)$ message. All correct processors therefore receive $b$ together with an L-notarisation (and an M-notarisation) for $b$ by $t+3\delta$, and also leave view $v$ by this time.  This establishes the claim of the lemma.
\end{proof}

\begin{lemma} \label{olem2} Suppose the first correct processor to enter view $v$ does so at $t\geq $GST. Then, whether or not $\mathtt{lead}(v)$ is correct,  all correct processors leave view $v$ by $t+O(\Delta)$.
\end{lemma}
\begin{proof} Suppose the first correct processor to enter view $v$ does so at $t\geq $ GST. Towards a contradiction, suppose some correct processor does not leave view $v$ by $t+2\Delta+3\delta$. As before, it follows that all correct processors enter view $v$ by $t+\delta$. By timeslot $t+\delta +2\Delta$, all correct processors have either voted for some view $v$ block, or else sent a nullify$(v)$ message. If any correct processor receives an M-notarisation for a view $v$ block by $t+2\Delta +2\delta$, then it forwards it on to all processors. This means all correct processors leave the view by $t+2\Delta +3\delta$, giving an immediate contradiction. So, suppose otherwise. If $p_j$ is a  correct processor that votes for a view $v$ block $b$, it follows that, by $t+2\Delta +2\delta$,  $p_j$  receives messages from at least $(n-f)-(2f)=n-3f\geq 2f+1$ processors, each of which is either:
\begin{enumerate}
\item[(i)]  A nullify$(v)$ message, or;
\item[(ii)]  A vote for a view $v$ block different than $b$.
\end{enumerate}
 So, the conditions of lines \ref{beginN}-\ref{endN} are satisfied at this time, meaning that $p_j$ sends a nullify$(v)$ message (line \ref{sendN}). Any correct processor that does not vote for a view $v$ block also sends a nullify$(v)$ message by this time. So, all correct processors receive a nullification for view $v$ by $t+ 2\Delta +3\delta$, giving the required contradiction.
\end{proof}

\begin{lemma}
    Minimmit is optimistically responsive.
\end{lemma}
\begin{proof}
    Suppose $\text{tr}$ is first received by a correct processor at $t\geq $ GST. Since we assume correct processors send new transactions to all other processors upon first receiving them, $\text{tr}$ is received by all correct processors by $t+\delta$. Let $v_0$ be the greatest view that any correct processor is in at $t+\delta$, and let $v_1$ be the least view $>v_0$ such that $\mathtt{lead}(v)$ is correct. From Lemmas \ref{olem1} and \ref{olem2}, it follows that all correct processors enter view $v_1$ by time $t+O(f_a\Delta +\delta)$, and that all correct processors also finalise a view $v_1$ block, $b$ say,  by $t+O(f_a\Delta +\delta)$. According to the definition of the procedure ProposeChild$(b,v)$, $\text{tr}$ will be included in an ancestor of $b$. Since all ancestors of $b'$ of $b$ must receive M-notarisations prior to $\mathtt{lead}(v)$ proposing $b$,  at least $f+1$ correct processors send each such $b'$ to all processors, and correct processors receive all ancestors of $b$ by $t+O(f_a\Delta +\delta)$. All correct processors therefore finalise $\text{tr}$ by time $t+O(f_a\Delta +\delta)$.
\end{proof}

%
\section{Optimisations} \label{opt}
In Section \ref{formal}, we gave a specification aimed at simplicity. In this section, we describe a number of possible optimisations.

\subsection{Progression through views} The specification of Section \ref{formal} requires correct processors to progress sequentially through views. To recover quickly from periods of asynchrony, one can allow a correct processor that is presently in view $v'$ to progress immediately to view $v+1>v'$ upon seeing a nullification for view $v$, or an M-notarisation for view $v$. This requires the following modifications:
\begin{itemize}
\item[(1)] If $v''<v$ and a  correct processor $p_i$ in view $v$ receives an M-notarisation for some view $v''$ block $b$, and if $p_i$ has not voted for any view $v''$ block and has not sent a nullify$(v'')$ message, then $p_i$ must vote for $b$. This is now required to ensure that correct leaders finalise new blocks after GST, i.e., that a correct leader's block receives an L-notarisation.
\item[(2)] Upon entering view $v$, and before running  $\text{ProposeChild}(\text{SelectParent}(\mathtt{S},v),v)$, $\mathtt{lead}(v)$ must now wait until there exists $v''<v$ such that it has received:
\begin{itemize}
    \item An M-notarization for some view $v''$ block, and;
    \item Nullifications for all views in $(v'',v)$.
\end{itemize}
Lemmas \ref{L1} and \ref{olem1} still go through with this change in place, since, if the first correct processor to enter the view does so at $t\geq$ GST, then  $\mathtt{lead}(v)$ will still receive these required messages by $t+\delta$.
\end{itemize}
With these changes in place, the proofs of Sections \ref{consec}--\ref{resp} go through almost unchanged.

\subsection{Reducing the size of votes} \label{reduce}
In Section \ref{formal}, votes take the form $(\text{vote},b)$, and so include the entire block $b$.
 For constant-sized blocks, this does not affect asymptotic communication complexity. However, when blocks are large, a standard optimisation is to use votes of the form $(\text{vote},H(b))$, containing only the block's hash.

This optimisation introduces a data availability challenge, which is common to all protocols using votes that only specify the block's hash: a Byzantine leader might propose a block $b$
that receives sufficient votes for finalisation, but fail to send $b$
itself to all correct processors. Since correct processors need the actual block content to update their logs, they must have a mechanism to retrieve missing blocks.

A standard solution exploits the fact that finalisation requires votes from many processors. If a block is finalised, at least $n-f$ processors (including many correct ones) must have received and voted for it. Correct processors can therefore use a (potentially rate-limited) "pull" mechanism to retrieve any missing finalised blocks from peers who possess them, ensuring data availability without relying on potentially Byzantine leaders. An alternative approach, described in Section \ref{eras}, is to have leaders  broadcast blocks using erasure coding techniques.

\subsection{Erasure coding} \label{eras}
 In Minimmit, as in all leader-based protocols, leaders must broadcast potentially large blocks to all $n$ processors. With large blocks or high transaction throughput, leader bandwidth can become a bottleneck, limiting overall system performance.
 Following approaches used in recent SMR protocols \cite{shoup2023sing,shoup2025kudzu,alpen}, we can apply erasure coding techniques \cite{alhaddad2021succinct,cachin2005asynchronous} to distribute the communication load. The leader encodes each block into $n$ fragments such that any $d$ fragments suffice to reconstruct the original block (where $d$ is a parameter, set as required).
 
 As explained in Appendix \ref{erasureintu}, a \emph{naive} approach to the integration of erasure codes causes the resulting version of Minimmit to have a \emph{data expansion} rate of 5: If  blocks are of size $B$, this means the total amount of block data sent by each leader (to all processors combined) is approximately $5B$. The factor 5 arises because an M-notarisation only suffices to ensure that $f+1$ correct processors have received their corresponding block fragments, meaning that $f+1$ out of $n$ fragments should suffice for block reconstruction.  However, in Appendices \ref{erasureintu}-\ref{erasureanalysis}, we also describe how to integrate erasure codes with Minimmit to give a data expansion rate of 2.5.

\vspace{0.1cm}
\noindent \textbf{Trade-offs}. This optimisation reduces leader bandwidth requirements but adds fragment verification overhead. The approach is most beneficial when block sizes are large relative to network capacity.
As noted in Section \ref{reduce}, a benefit of the approach is that it ensures correct processors receive all finalised blocks, without requiring the use of a "pull" mechanism to retrieve any missing finalised blocks from peers who possess them.

\subsection{Threshold signatures and communication complexity} Section \ref{formal} assumes correct processors send newly received transactions to all others. In practice, transactions are typically disseminated through gossip networks, where each processor forwards transactions to a small constant number of peers. If transactions are constant-bounded in size, this approach maintains constant communication overhead per processor per transaction. Alternative dissemination mechanisms like Narwhal \cite{danezis2022narwhal} can also be employed. In our analysis here, we treat transaction propagation as a black box and focus on the consensus layer.  We assume blocks have constant-bounded size and, following standard practice, analyse only the communication costs required for consensus, taking mempool formation as given.

\vspace{0.2cm}
\noindent \textbf{Message complexity}. Per view, the protocol requires:
\begin{itemize}
\item[-] Leader proposal: $O(n)$ messages.
\item[-] Votes: $O(n^2)$ messages (each processor sends $\leq 1$ vote to all others).
\item[-] Nullify messages: $O(n^2)$ messages.
\item[-] Forwarding notarisations/nullifications: $O(n^2)$ messages.
\end{itemize}

\vspace{0.2cm}
\noindent \textbf{Communication complexity}. Each notarisation and nullification contains $\Omega(n)$ signatures, resulting in communication complexity greater than $O(n^2)$ per view.  The standard approach of using threshold signatures \cite{boneh2001short,shoup2000practical} can be used to ensure that the communication complexity per view is $O(n^2)$. We define:
\begin{itemize}
\item An `M-certificate' for $b$ to be threshold signature (of constant-bounded length for a given security parameter) formed from $2f+1$ votes for $b$ by different processors.

\item An `N-certificate' for view $v$, to be a threshold signature formed from $2f+1$ nullify$(v)$ messages by different processors.
\end{itemize}

\vspace{0.1cm}
\noindent \emph{Threshold signature implemetation details}. If $p_i$ has not already formed or received an M-certificate for $b$, then, upon receipt of $2f+1$ votes for $b$ by different processors, $p_i$ combines the $2f+1$ signatures to form an M-certificate (a single threshold signature). Rather than storing and forwarding the $2f+1$ votes, $p_i$ stores and forwards the M-certificate.
Similarly, $p_i$ stores and sends N-certificates, rather than storing and sending large collections of nullify$(v)$ messages.
Since partial signatures for any processor can be derived from $2f+1$ partial signatures, we must now stipulate that processors finalise a block $b$ upon receiving $n-f$ votes for $b$ directly from the corresponding processors.

We note that the pseudocode specified in Algorithm 1 does not require processors to forward L-notarisations. If one wished to implement threshold signatures also for L-notarisations (requiring the different threshold $n-f$), then one would need to establish two separate threshold signature schemes (two shared secrets), i.e., we require a separate shared secret for each threshold. This also means that each vote requires two signatures: one corresponding to each threshold. Of course, these two signatures can be computed/verified in parallel. Verification for the two signatures can also be transformed into an aggregate signature verification because the two signatures are over the same message payload.

\vspace{0.2cm}
\noindent \textbf{Alternative approaches}. We note that some protocols (e.g., Hotstuff \cite{yin2019hotstuff}) achieve linear communication complexity per view by relaying all messages via the leader. However, this approach significantly increases the number of communication rounds required, and the leader anyway acts as a communication bottleneck (e.g., see \cite{lewis2025pipes}).

\subsection{Compressed nullifications}
\noindent \textbf{The problem: nullification build-up during asynchrony}. During extended periods of asynchrony, processors may generate nullifications for many consecutive views without being able to finalise new blocks. When synchrony is restored, correct processors must exchange all accumulated nullifications before they can vote for new proposals. Since a processor requires nullifications for all intermediate views between a block's parent and the current view (see Section \ref{formal}), this can create substantial communication overhead.

\vspace{0.2cm}
\noindent \textbf{Solution: nullification aggregation}. For signature schemes supporting aggregation (e.g., BLS \cite{boneh2001short}), we can compress consecutive nullifications. Given threshold signatures for views in the range $[v, v']$, we aggregate them into a single signature $\sigma$ of constant size and send the tuple $(v, v', \sigma)$.

\vspace{0.2cm}
\noindent \textbf{Verification}. Aggregate signature verification requires specifying the message sequence and corresponding public keys. Here, the message sequence is implicitly defined by the view range $[v, v']$, and each nullification uses the same shared public key from the threshold scheme.



\section{Experiments} \label{exp}
We test Minimmit against Simplex~\cite{chan2023simplex} and Kudzu~\cite{shoup2025kudzu} using a deterministic simulator that executes protocol specifications on configurable network topologies. The simulator implementation and experiment configurations are released under both MIT and Apache-2 licenses.\footnote{\url{https://github.com/commonwarexyz/monorepo/tree/19f19d32760daf1d497295726ec92a1e6b84959f/examples/estimator}} A step-by-step guide for recreating every experiment appears in Appendix~\ref{app:repro}.

\vspace{0.2cm}
\noindent \textbf{Protocol selection}. Simplex introduced state-of-the-art transaction latency among 3-round finality protocols. Its design inspired Minimmit: both protocols decouple view iteration from finalisation, shrinking the worst-case delay before a transaction is included in a block. This structural similarity makes it a natural baseline against which to quantity the impact of Minimmit's relaxed Byzantine fault tolerance.

Kudzu (and the similarly constructed Alpenglow~\cite{alpen}) recently delivered state-of-the-art transaction latency for 2-round finality protocols by concurrently evaluating fast and slow paths. Unlike Alpenglow, which incorporates a scheme for disseminating erasure-coded block data during fixed 400ms slots\footnote{A processor in Votor will not cast a vote for some block until it has recovered the entire block from Rotor.}, Kudzu is responsive and serves as a better candidate for comparison to Minimmit. We defer a comprehensive comparison to Alpenglow's block dissemination to future work. 

\vspace{0.2cm}
\noindent \textbf{Simulation procedure}. Each simulation iteration designates a leader and treats the remaining processors as replicas. The leader initiates the view by broadcasting a payload of configurable size to all replicas. Replicas process received payloads after at most a configurable number of pending messages and broadcast their own protocol-specific messages. The simulator records the time at which each processor reaches salient protocol milestones and reports the mean and standard deviation over all processors. Because the runtime is deterministic, these measurements are reproducible by re-running the same command.\footnote{\url{https://github.com/commonwarexyz/monorepo/tree/19f19d32760daf1d497295726ec92a1e6b84959f/runtime/src/deterministic.rs}}

\vspace{0.2cm}
\noindent \textbf{Latency and bandwidth model}. For every message transmission we sample a delay from a normal distribution with mean equal to the p50\footnote{\url{https://www.cloudping.co/api/latencies?percentile=p_50&timeframe=1Y}} AWS inter-region latency and standard deviation $(\text{p90}-\text{p50})$.\footnote{\url{https://www.cloudping.co/api/latencies?percentile=p_90&timeframe=1Y}} Each processor has symmetric 1~Gbps (125{,}000{,}000~B/s) ingress and egress budgets, so message delivery time is the sum of the sampled network delay and the bandwidth-limited transmission time. We assume linear bandwidth usage and max-min fairness for bandwidth allocation, yielding conservative latency estimates. In practice, burstable bandwidth and traffic prioritisation provided by cloud operators would only decrease the reported latencies.

\vspace{0.2cm}
\noindent \textbf{Uniform global deployment}. We place 50 processors uniformly across ten AWS regions (us-west-1, us-east-1, eu-west-1, ap-northeast-1, eu-north-1, ap-south-1, sa-east-1, eu-central-1, ap-northeast-2, ap-southeast-2) and enable the ``reducing the size of votes'' optimisation from Section \ref{reduce} for all protocols, so that the leader distributes full blocks (i.e. no erasure coding) while replicas vote on digests. We assume instantaneous block production and begin broadcasting at time $t=0$. Table~\ref{tab:global} summarises the resulting latency when proposing 32~KB blocks, the block size at which Minimmit processes 1,000 transactions per second\footnote{Assumes each transaction is 200~B.} in this configuration.

\begin{table}[h]
  \centering
  \caption{Uniform global deployment (50 processors, symmetric 1~Gbps links).}
  \label{tab:global}
  \begin{tabular}{|l|c|c|c|}
    \hline
    \textbf{Protocol} & \textbf{View Latency} & \textbf{Block Latency} & \textbf{Transaction Latency} \\
    \hline
    Simplex & $194.61 \pm 30.34$~ms & $299.34 \pm 25.98$~ms & $493.95 \pm 7.50$~ms \\
    Kudzu & $189.94 \pm 29.32$~ms & $220.31 \pm 28.58$~ms & $410.25 \pm 7.61$~ms \\
     Minimmit & $146.07 \pm 21.33$~ms & $220.3 \pm 28.57$~ms & $366.37 \pm 7.06$~ms \\
    \hline
  \end{tabular}
\end{table}

Minimmit reduces transaction latency by 25.8\% relative to Simplex and by 10.7\% relative to Kudzu. The improvement stems from Minimmit's ability to progress views after collecting $2f+1$ votes.

\vspace{0.2cm}
\noindent \textbf{Region-centric deployment}. We next cluster 25 processors in the United States (13 in us-west-1 and 12 in us-east-1) and place the remaining 25 uniformly across the other eight regions. Table~\ref{tab:regional} reports the resulting latencies for the same 32~KB blocks, increasing Minimmit's effective processing rate from 1,000 to 1,500 transactions per second at a lower transaction  latency.

\begin{table}[h]
  \centering
  \caption{Region-centric deployment (50 processors, symmetric 1~Gbps links).}
  \label{tab:regional}
  \begin{tabular}{|l|c|c|c|}
    \hline
    \textbf{Protocol} & \textbf{View Latency} & \textbf{Block Latency} & \textbf{Transaction Latency} \\
    \hline
    Simplex & $149.95 \pm 24.58$~ms & $222.32 \pm 24.05$~ms & $372.27 \pm 6.97$~ms \\
    Kudzu & $139.77 \pm 26.17$~ms & $185.16 \pm 24.09$~ms & $324.93 \pm 7.09$~ms \\
     Minimmit & $104.93 \pm 34.71$~ms & $187.67 \pm 27.13$~ms & $292.6 \pm 7.86$~ms \\
    \hline
  \end{tabular}
\end{table}

Under this regional skew, Minimmit decreases transaction latency by 21.4\% relative to Simplex and by 9.95\% relative to Kudzu. Kudzu attains slightly lower block latency because its slow path races the fast path, and in this topology some leaders complete two rounds gathering $3f+1$ votes before Minimmit collects $n-f$ votes.

\vspace{0.2cm}
\noindent \textbf{Large blocks}. Larger blocks take longer to transmit: broadcasting a 1~MB block to 50 processors over 1~Gbps links already requires roughly 400~ms of egress latency.\footnote{Recall, we assume bandwidth allocation is max-min fair.} Table~\ref{tab:global-large} reproduces the global deployment with 1~MB blocks, the block size at which Minimmit processes 10,000 transactions per second\footnote{Assumes each transaction is 200~B.} in this configuration.

\begin{table}[h]
  \centering
  \caption{Uniform global deployment (50 processors, symmetric 1~Gbps links).}
  \label{tab:global-large}
  \begin{tabular}{|l|c|c|c|}
    \hline
    \textbf{Protocol} & \textbf{View Latency} & \textbf{Block Latency} & \textbf{Transaction Latency} \\
    \hline
    Simplex & $593.61 \pm 30.34$ms & $698.34 \pm 25.98$~ms & $1291.95 \pm 7.50$ms \\
    Kudzu & $588.94 \pm 29.32$~ms & $619.31 \pm 28.58$~ms & $1208.25 \pm 7.61$~ms \\
     Minimmit & $545.07 \pm 21.33$~ms & $619.3 \pm 28.57$~ms & $1164.37 \pm 7.06$~ms \\
    \hline
  \end{tabular}
\end{table}

 Introducing a Reed--Solomon erasure coding scheme reduces this bottleneck. We split each block into 50 fragments, have the leader broadcast a fragment to every processor, and require replicas to broadcast their fragment to all other processors when casting a vote. The fragment size is determined by the quorum needed for view progression: Simplex targets $3f+1$ replicas and therefore transmits 61.72~KB fragments so that any $f+1$ fragments suffice for reconstruction; E-Minimmit (as described in Appendices \ref{erasuresetup}-\ref{erasureanalysis} targets $5f+1$ replicas and sends 55.19~KB fragments to ensure $2f+1$ fragments suffice; similarly, Kudzu also targets $5f+1$ replicas and requires 55.19~KB fragments to guarantee reconstruction from $2f+1$ fragments. Table~\ref{tab:global-ec} reports the resulting latencies, increasing Minimmit's effective processing rate from 10,000 to around 30,000 transactions per second at a lower transaction latency.

\begin{table}[h]
  \centering
  \caption{Uniform global deployment with erasure coding (1~MB blocks).}
  \label{tab:global-ec}
  \begin{tabular}{|l|c|c|c|}
    \hline
    \textbf{Protocol} & \textbf{View Latency} & \textbf{Block Latency} & \textbf{Transaction Latency} \\
    \hline
    Simplex & $230.62 \pm 30.33$~ms & $335.36 \pm 25.97$~ms & $565.98 \pm 7.50$~ms \\
    Kudzu & $220.94 \pm 29.33$~ms & $251.29 \pm 28.63$~ms & $472.23 \pm 7.61$~ms \\
    Minimmit & $177.06 \pm 21.45$~ms & $251.29 \pm 28.63$~ms & $428.35 \pm 7.08$~ms \\
    \hline
  \end{tabular}
\end{table}

With coding enabled, Minimmit achieves the lowest view latency, block latency and  transaction latency.  Minimmit \emph{maximises throughput} across all protocols, since blocks of a fixed size are produced at a faster rate.

Increasing per processor bandwidth from 1~Gbps to 10~Gbps under this naïve coding scheme, produces latencies illustrated in Table~\ref{tab:global-ec-ib}.  

\begin{table}[h]
  \centering
  \caption{Uniform global deployment with increased bandwidth and erasure coding (10~Gbps per processor and 1~MB blocks).}
  \label{tab:global-ec-ib}
  \begin{tabular}{|l|c|c|c|}
    \hline
    \textbf{Protocol} & \textbf{View Latency} & \textbf{Block Latency} & \textbf{Transaction Latency} \\
    \hline
    Simplex & $186.63 \pm 30.35$~ms & $291.33 \pm 26.0$~ms & $477.96 \pm 7.51$~ms \\
    Kudzu & $181.86 \pm 29.31$~ms & $212.3 \pm 28.6$~ms & $394.16 \pm 7.61$~ms \\
    Minimmit & $138.02 \pm 21.37$~ms & $212.3 \pm 28.6$~ms & $350.32\pm 7.07$~ms \\
    \hline
  \end{tabular}
\end{table}

\noindent  Minimmit decreases transaction latency by 26.7\% relative to Simplex and 11.1\% relative to Kudzu, and also decreases view latency by 26.0\% relative to Simplex and 24.1\% relative to Kudzu. 
 

%
\section{Related work} \label{rw}

\noindent \textbf{Classical Byzantine Consensus}. The study of protocols for reaching consensus in the presence of Byzantine faults was introduced by Lamport, Shostak, and Pease \cite{lamport1982byzantine}. Dwork, Lynch and Stockmeyer \cite{DLS88} showed that $n\geq 3f+1$ is optimal for partial synchrony.
Standard protocols using the assumption $n\geq 3f+1$,  such as PBFT \cite{castro1999practical} and Tendermint \cite{buchman2016tendermint,buchman2018latest}, satisfy 3-round finality. As shown by \cite{abraham2021good}, this is optimal.

\vspace{0.2cm}
\noindent \textbf{Optimistic Responsiveness}. While many standard protocols, such as PBFT, satisfy forms of optimistic responsiveness, a specific form of the concept was first discussed in \cite{pass2018thunderella}. Optimistic responsiveness has been further studied in a number of papers (e.g., \cite{yin2019hotstuff,abraham2020sync}) and can be defined in a number of ways  \cite{yin2019hotstuff,lewis2023permissionless,lewis2024lumiere}.

\vspace{0.2cm}
\noindent \textbf{Fast-Path Approaches}.
A long line of work \cite{brasileiro2001consensus,friedman2005simple,guerraoui2007refined,kursawe2002optimistic,martin2006fast,song2008bosco} considers protocols with a `fast path', which allows for quick termination/finalisation in certain `good' scenarios.
Kursawe \cite{kursawe2002optimistic} describes an agreement protocol that runs with $3f+1$ processors, but is able to commit in two steps  when all processes act correctly  and the network is synchronous, falling back  to a randomised asynchronous consensus protocol otherwise. FaB \cite{martin2006fast} extends Kursawe's approach by more closely integrating the fast path and the fall-back mechanism (the `slow path'), and by introducing a parameterised model with $n \geq 3f + 2p + 1$ processors. Fast termination is achieved, so long as at most $p$ processors are Byzantine. Unfortunately, the protocol suffers from a liveness bug \cite{abraham2017revisiting}. FaB also gives a proof that the assumed bound $n\geq 3f+ 2p+ 1$  is tight. However, it was pointed out in \cite{kuznetsov2021revisiting} and \cite{abraham2021good} that the proof only applies to a specific form of protocol. Kuznetsov et al.~\cite{kuznetsov2021revisiting} show that, in fact,  the bound $n\geq 3f+2p-1$ is optimal. Zyzzyva \cite{kotla2007zyzzyva} also builds on FaB by (similarly) integrating the fast and slow paths, and by describing an SMR protocol, rather than a protocol for one-shot consensus. As pointed out in \cite{abraham2017revisiting}, the view-change mechanism in Zyzzyva does not guarantee safety when leaders are faulty. SBFT \cite{gueta2019sbft} also builds on the ideas introduced in FaB in order to allow the fast path to  tolerate a small number of crash failures.

\vspace{0.2cm}
\noindent \textbf{Modern 2-Round Protocols}.
Alpenglow \cite{alpen} is formally analysed under the assumption that $n\geq 5f+1$. The paper also considers circumstances in which the protocol can tolerate a further $f$ crash failures, but the required assumptions for this case (essentially that Byzantine leaders cannot carry out a form of proposal equivocation) do not hold under partial synchrony.  Banyan \cite{vonlanthen2024banyan} carries out the fast path in parallel with the slow path mechanism, but can suffer from unbounded message complexity with faulty leaders. Kudzu (like FaB) makes the more general assumption that $n\geq 3f+2p+1$ for a tunable parameter $p$. Kudzu and Alpenglow also describe the use of erasure coding techniques to allow for significantly improved maximum throughput (such techniques were also employed by SMR protocols in earlier papers, such as \cite{shoup2023sing}).

 Compared to Minimmit, Alpenglow, and Kudzu, Hydrangea has improved resilience to crash failures.  For a parameter $k\geq 0$, and for a system of $n = 3f + 2c + k + 1$ processors, Hydrangea achieves 2-round finality, so long as the number of faulty processors (Byzantine or crash) is at most $p = \lfloor \frac{c+k}{2} \rfloor$. In the case that $c=0$, this aligns precisely with the bounds provided by Kudzu. However, in more adversarial settings with up to $f$ Byzantine faults and $c$ crash faults,
Hydrangea also obtains finality after two rounds of voting.

As for Minimmit, ChonkyBFT \cite{francca2025chonkybft} assumes $n\geq 5f+1$ and employs a single round of voting, but does not have Minimmit's mechanism for fast view progression.

\vspace{0.2cm}
\noindent \textbf{Positioning of Minimmit}. Minimmit assumes $n\geq 5f+1$ and achieves 2-round finality. The advantage of Minimmit over all previous approaches to 2-round finality is its fast view change mechanism, which, as described in Sections \ref{intro} and \ref{exp}, allows for decreased view and transaction latency. While subjective, we also believe that the simplicity of the protocol will make it attractive to practitioners.
\section{Final comments} \label{fc}
We have presented Minimmit, a Byzantine fault-tolerant SMR protocol that achieves reduced transaction latency through a novel view-change mechanism. By decoupling view progression from transaction finality—requiring only $2f+1$ votes for view changes while requiring $n-f$ votes for finalisation—Minimmit demonstrates that significant latency improvements are possible without sacrificing safety or liveness guarantees.

Our experimental evaluation shows an approximately 23\% reduction in view latency and an 11\% reduction in transaction latency compared to existing approaches, achieved through faster view progression in geographically distributed networks. The protocol's simplicity, requiring no complex slow-path mechanisms, may facilitate practical adoption in systems where low latency is critical.

\section{Acknowledgements}
We would like to thank Ittai Abraham, Benjamin Chan, Denis Kolegov, Ling Ren, and Victor Shoup for helpful conversations.

\bibliographystyle{ACM-Reference-Format}

\begin{thebibliography}{34}


\ifx \showCODEN    \undefined \def \showCODEN     #1{\unskip}     \fi
\ifx \showDOI      \undefined \def \showDOI       #1{#1}\fi
\ifx \showISBNx    \undefined \def \showISBNx     #1{\unskip}     \fi
\ifx \showISBNxiii \undefined \def \showISBNxiii  #1{\unskip}     \fi
\ifx \showISSN     \undefined \def \showISSN      #1{\unskip}     \fi
\ifx \showLCCN     \undefined \def \showLCCN      #1{\unskip}     \fi
\ifx \shownote     \undefined \def \shownote      #1{#1}          \fi
\ifx \showarticletitle \undefined \def \showarticletitle #1{#1}   \fi
\ifx \showURL      \undefined \def \showURL       {\relax}        \fi
\providecommand\bibfield[2]{#2}
\providecommand\bibinfo[2]{#2}
\providecommand\natexlab[1]{#1}
\providecommand\showeprint[2][]{arXiv:#2}

\bibitem[Abraham et~al\mbox{.}(2017)]%
        {abraham2017revisiting}
\bibfield{author}{\bibinfo{person}{Ittai Abraham}, \bibinfo{person}{Guy Gueta},
  \bibinfo{person}{Dahlia Malkhi}, \bibinfo{person}{Lorenzo Alvisi},
  \bibinfo{person}{Rama Kotla}, {and} \bibinfo{person}{Jean-Philippe Martin}.}
  \bibinfo{year}{2017}\natexlab{}.
\newblock \showarticletitle{Revisiting fast practical byzantine fault
  tolerance}.
\newblock \bibinfo{journal}{\emph{arXiv preprint arXiv:1712.01367}}
  (\bibinfo{year}{2017}).
\newblock


\bibitem[Abraham et~al\mbox{.}(2020)]%
        {abraham2020sync}
\bibfield{author}{\bibinfo{person}{Ittai Abraham}, \bibinfo{person}{Dahlia
  Malkhi}, \bibinfo{person}{Kartik Nayak}, \bibinfo{person}{Ling Ren}, {and}
  \bibinfo{person}{Maofan Yin}.} \bibinfo{year}{2020}\natexlab{}.
\newblock \showarticletitle{Sync hotstuff: Simple and practical synchronous
  state machine replication}. In \bibinfo{booktitle}{\emph{2020 IEEE Symposium
  on Security and Privacy (SP)}}. IEEE, \bibinfo{pages}{106--118}.
\newblock


\bibitem[Abraham et~al\mbox{.}(2021)]%
        {abraham2021good}
\bibfield{author}{\bibinfo{person}{Ittai Abraham}, \bibinfo{person}{Kartik
  Nayak}, \bibinfo{person}{Ling Ren}, {and} \bibinfo{person}{Zhuolun Xiang}.}
  \bibinfo{year}{2021}\natexlab{}.
\newblock \showarticletitle{Good-case latency of byzantine broadcast: A
  complete categorization}. In \bibinfo{booktitle}{\emph{Proceedings of the
  2021 ACM Symposium on Principles of Distributed Computing}}.
  \bibinfo{pages}{331--341}.
\newblock


\bibitem[Alhaddad et~al\mbox{.}(2021)]%
        {alhaddad2021succinct}
\bibfield{author}{\bibinfo{person}{Nicolas Alhaddad}, \bibinfo{person}{Sisi
  Duan}, \bibinfo{person}{Mayank Varia}, {and} \bibinfo{person}{Haibin Zhang}.}
  \bibinfo{year}{2021}\natexlab{}.
\newblock \showarticletitle{Succinct erasure coding proof systems}.
\newblock \bibinfo{journal}{\emph{Cryptology ePrint Archive}}
  (\bibinfo{year}{2021}).
\newblock


\bibitem[Boneh et~al\mbox{.}(2001)]%
        {boneh2001short}
\bibfield{author}{\bibinfo{person}{Dan Boneh}, \bibinfo{person}{Ben Lynn},
  {and} \bibinfo{person}{Hovav Shacham}.} \bibinfo{year}{2001}\natexlab{}.
\newblock \showarticletitle{Short signatures from the Weil pairing}. In
  \bibinfo{booktitle}{\emph{International conference on the theory and
  application of cryptology and information security}}. Springer,
  \bibinfo{pages}{514--532}.
\newblock


\bibitem[Brasileiro et~al\mbox{.}(2001)]%
        {brasileiro2001consensus}
\bibfield{author}{\bibinfo{person}{Francisco Brasileiro},
  \bibinfo{person}{Fab{\'\i}ola Greve}, \bibinfo{person}{Achour
  Most{\'e}faoui}, {and} \bibinfo{person}{Michel Raynal}.}
  \bibinfo{year}{2001}\natexlab{}.
\newblock \showarticletitle{Consensus in one communication step}. In
  \bibinfo{booktitle}{\emph{International Conference on Parallel Computing
  Technologies}}. Springer, \bibinfo{pages}{42--50}.
\newblock


\bibitem[Buchman(2016)]%
        {buchman2016tendermint}
\bibfield{author}{\bibinfo{person}{Ethan Buchman}.}
  \bibinfo{year}{2016}\natexlab{}.
\newblock \emph{\bibinfo{title}{Tendermint: Byzantine fault tolerance in the
  age of blockchains}}.
\newblock \bibinfo{thesistype}{Ph.\,D. Dissertation}.
\newblock


\bibitem[Buchman et~al\mbox{.}(2018)]%
        {buchman2018latest}
\bibfield{author}{\bibinfo{person}{Ethan Buchman}, \bibinfo{person}{Jae Kwon},
  {and} \bibinfo{person}{Zarko Milosevic}.} \bibinfo{year}{2018}\natexlab{}.
\newblock \showarticletitle{The latest gossip on BFT consensus}.
\newblock \bibinfo{journal}{\emph{arXiv preprint arXiv:1807.04938}}
  (\bibinfo{year}{2018}).
\newblock


\bibitem[Cachin and Tessaro(2005)]%
        {cachin2005asynchronous}
\bibfield{author}{\bibinfo{person}{Christian Cachin} {and}
  \bibinfo{person}{Stefano Tessaro}.} \bibinfo{year}{2005}\natexlab{}.
\newblock \showarticletitle{Asynchronous verifiable information dispersal}. In
  \bibinfo{booktitle}{\emph{24th IEEE Symposium on Reliable Distributed Systems
  (SRDS'05)}}. IEEE, \bibinfo{pages}{191--201}.
\newblock


\bibitem[Castro et~al\mbox{.}(1999)]%
        {castro1999practical}
\bibfield{author}{\bibinfo{person}{Miguel Castro}, \bibinfo{person}{Barbara
  Liskov}, {et~al\mbox{.}}} \bibinfo{year}{1999}\natexlab{}.
\newblock \showarticletitle{Practical byzantine fault tolerance}. In
  \bibinfo{booktitle}{\emph{OsDI}}, Vol.~\bibinfo{volume}{99}.
  \bibinfo{pages}{173--186}.
\newblock


\bibitem[Chan and Pass(2023)]%
        {chan2023simplex}
\bibfield{author}{\bibinfo{person}{Benjamin~Y Chan} {and}
  \bibinfo{person}{Rafael Pass}.} \bibinfo{year}{2023}\natexlab{}.
\newblock \showarticletitle{Simplex consensus: A simple and fast consensus
  protocol}. In \bibinfo{booktitle}{\emph{Theory of Cryptography Conference}}.
  Springer, \bibinfo{pages}{452--479}.
\newblock


\bibitem[Danezis et~al\mbox{.}(2022)]%
        {danezis2022narwhal}
\bibfield{author}{\bibinfo{person}{George Danezis}, \bibinfo{person}{Lefteris
  Kokoris-Kogias}, \bibinfo{person}{Alberto Sonnino}, {and}
  \bibinfo{person}{Alexander Spiegelman}.} \bibinfo{year}{2022}\natexlab{}.
\newblock \showarticletitle{Narwhal and tusk: a dag-based mempool and efficient
  bft consensus}. In \bibinfo{booktitle}{\emph{Proceedings of the Seventeenth
  European Conference on Computer Systems}}. \bibinfo{pages}{34--50}.
\newblock


\bibitem[Dwork et~al\mbox{.}(1988)]%
        {DLS88}
\bibfield{author}{\bibinfo{person}{Cynthia Dwork}, \bibinfo{person}{Nancy~A.
  Lynch}, {and} \bibinfo{person}{Larry Stockmeyer}.}
  \bibinfo{year}{1988}\natexlab{}.
\newblock \showarticletitle{Consensus in the Presence of Partial Synchrony}.
\newblock \bibinfo{journal}{\emph{J. ACM}} \bibinfo{volume}{35},
  \bibinfo{number}{2} (\bibinfo{year}{1988}), \bibinfo{pages}{288--323}.
\newblock


\bibitem[Fran{\c{c}}a et~al\mbox{.}(2025)]%
        {francca2025chonkybft}
\bibfield{author}{\bibinfo{person}{Bruno Fran{\c{c}}a}, \bibinfo{person}{Denis
  Kolegov}, \bibinfo{person}{Igor Konnov}, {and} \bibinfo{person}{Grzegorz
  Prusak}.} \bibinfo{year}{2025}\natexlab{}.
\newblock \showarticletitle{ChonkyBFT: Consensus Protocol of ZKsync}.
\newblock \bibinfo{journal}{\emph{arXiv preprint arXiv:2503.15380}}
  (\bibinfo{year}{2025}).
\newblock


\bibitem[Friedman et~al\mbox{.}(2005)]%
        {friedman2005simple}
\bibfield{author}{\bibinfo{person}{Roy Friedman}, \bibinfo{person}{Achour
  Mostefaoui}, {and} \bibinfo{person}{Michel Raynal}.}
  \bibinfo{year}{2005}\natexlab{}.
\newblock \showarticletitle{Simple and efficient oracle-based consensus
  protocols for asynchronous Byzantine systems}.
\newblock \bibinfo{journal}{\emph{IEEE Transactions on Dependable and Secure
  Computing}} \bibinfo{volume}{2}, \bibinfo{number}{1} (\bibinfo{year}{2005}),
  \bibinfo{pages}{46--56}.
\newblock


\bibitem[Guerraoui and Vukoli{\'c}(2007)]%
        {guerraoui2007refined}
\bibfield{author}{\bibinfo{person}{Rachid Guerraoui} {and}
  \bibinfo{person}{Marko Vukoli{\'c}}.} \bibinfo{year}{2007}\natexlab{}.
\newblock \showarticletitle{Refined quorum systems}. In
  \bibinfo{booktitle}{\emph{Proceedings of the twenty-sixth annual ACM
  symposium on Principles of distributed computing}}.
  \bibinfo{pages}{119--128}.
\newblock


\bibitem[Gueta et~al\mbox{.}(2019)]%
        {gueta2019sbft}
\bibfield{author}{\bibinfo{person}{Guy~Golan Gueta}, \bibinfo{person}{Ittai
  Abraham}, \bibinfo{person}{Shelly Grossman}, \bibinfo{person}{Dahlia Malkhi},
  \bibinfo{person}{Benny Pinkas}, \bibinfo{person}{Michael Reiter},
  \bibinfo{person}{Dragos-Adrian Seredinschi}, \bibinfo{person}{Orr Tamir},
  {and} \bibinfo{person}{Alin Tomescu}.} \bibinfo{year}{2019}\natexlab{}.
\newblock \showarticletitle{SBFT: A scalable and decentralized trust
  infrastructure}. In \bibinfo{booktitle}{\emph{2019 49th Annual IEEE/IFIP
  international conference on dependable systems and networks (DSN)}}. IEEE,
  \bibinfo{pages}{568--580}.
\newblock


\bibitem[Kniep et~al\mbox{.}(2025)]%
        {alpen}
\bibfield{author}{\bibinfo{person}{Quentin Kniep}, \bibinfo{person}{Jakub
  Sliwinski}, {and} \bibinfo{person}{Roger Wattenhofer}.}
  \bibinfo{year}{2025}\natexlab{}.
\newblock \showarticletitle{Solana Alpenglow Consensus}.
\newblock
  \bibinfo{journal}{\emph{\url{https://www.scribd.com/document/895233790/Solana-Alpenglow-White-Paper}}}
  (\bibinfo{year}{2025}).
\newblock


\bibitem[Kotla et~al\mbox{.}(2007)]%
        {kotla2007zyzzyva}
\bibfield{author}{\bibinfo{person}{Ramakrishna Kotla}, \bibinfo{person}{Lorenzo
  Alvisi}, \bibinfo{person}{Mike Dahlin}, \bibinfo{person}{Allen Clement},
  {and} \bibinfo{person}{Edmund Wong}.} \bibinfo{year}{2007}\natexlab{}.
\newblock \showarticletitle{Zyzzyva: speculative byzantine fault tolerance}. In
  \bibinfo{booktitle}{\emph{Proceedings of twenty-first ACM SIGOPS symposium on
  Operating systems principles}}. \bibinfo{pages}{45--58}.
\newblock


\bibitem[Kursawe(2002)]%
        {kursawe2002optimistic}
\bibfield{author}{\bibinfo{person}{Klaus Kursawe}.}
  \bibinfo{year}{2002}\natexlab{}.
\newblock \showarticletitle{Optimistic byzantine agreement}. In
  \bibinfo{booktitle}{\emph{21st IEEE Symposium on Reliable Distributed
  Systems, 2002. Proceedings.}} IEEE, \bibinfo{pages}{262--267}.
\newblock


\bibitem[Kuznetsov et~al\mbox{.}(2021)]%
        {kuznetsov2021revisiting}
\bibfield{author}{\bibinfo{person}{Petr Kuznetsov}, \bibinfo{person}{Andrei
  Tonkikh}, {and} \bibinfo{person}{Yan~X Zhang}.}
  \bibinfo{year}{2021}\natexlab{}.
\newblock \showarticletitle{Revisiting optimal resilience of fast byzantine
  consensus}. In \bibinfo{booktitle}{\emph{Proceedings of the 2021 ACM
  Symposium on Principles of Distributed Computing}}.
  \bibinfo{pages}{343--353}.
\newblock


\bibitem[Lamport et~al\mbox{.}(1982)]%
        {lamport1982byzantine}
\bibfield{author}{\bibinfo{person}{Leslie Lamport}, \bibinfo{person}{Robert
  Shostak}, {and} \bibinfo{person}{Marshall Pease}.}
  \bibinfo{year}{1982}\natexlab{}.
\newblock \showarticletitle{The Byzantine generals problem}.
\newblock \bibinfo{journal}{\emph{ACM Transactions on Programming Languages and
  Systems (TOPLAS)}} \bibinfo{volume}{4}, \bibinfo{number}{3}
  (\bibinfo{year}{1982}), \bibinfo{pages}{382--401}.
\newblock


\bibitem[Lewis-Pye et~al\mbox{.}(2024)]%
        {lewis2024lumiere}
\bibfield{author}{\bibinfo{person}{Andrew Lewis-Pye}, \bibinfo{person}{Dahlia
  Malkhi}, \bibinfo{person}{Oded Naor}, {and} \bibinfo{person}{Kartik Nayak}.}
  \bibinfo{year}{2024}\natexlab{}.
\newblock \showarticletitle{Lumiere: Making optimal bft for partial synchrony
  practical}. In \bibinfo{booktitle}{\emph{Proceedings of the 43rd ACM
  Symposium on Principles of Distributed Computing}}.
  \bibinfo{pages}{135--144}.
\newblock


\bibitem[Lewis-Pye et~al\mbox{.}(2025)]%
        {lewis2025pipes}
\bibfield{author}{\bibinfo{person}{Andrew Lewis-Pye}, \bibinfo{person}{Kartik
  Nayak}, {and} \bibinfo{person}{Nibesh Shrestha}.}
  \bibinfo{year}{2025}\natexlab{}.
\newblock \showarticletitle{The Pipes Model for Latency Analysis}.
\newblock \bibinfo{journal}{\emph{Cryptology ePrint Archive}}
  (\bibinfo{year}{2025}).
\newblock


\bibitem[Lewis-Pye and Roughgarden(2023)]%
        {lewis2023permissionless}
\bibfield{author}{\bibinfo{person}{Andrew Lewis-Pye} {and} \bibinfo{person}{Tim
  Roughgarden}.} \bibinfo{year}{2023}\natexlab{}.
\newblock \showarticletitle{Permissionless Consensus}.
\newblock \bibinfo{journal}{\emph{arXiv preprint arXiv:2304.14701}}
  (\bibinfo{year}{2023}).
\newblock

\bibitem[\protect\citeauthoryear{Lewis-Pye and Roughgarden}{Lewis-Pye and
  Roughgarden}{2025}]%
        {lewis2025beyond}
\bibfield{author}{\bibinfo{person}{Andrew Lewis-Pye} {and} \bibinfo{person}{Tim
  Roughgarden}.} \bibinfo{year}{2025}\natexlab{}.
\newblock \showarticletitle{Beyond optimal fault tolerance}.
\newblock \bibinfo{journal}{\emph{arXiv preprint arXiv:2501.06044}}
  (\bibinfo{year}{2025}).
\newblock


\bibitem[Martin and Alvisi(2006)]%
        {martin2006fast}
\bibfield{author}{\bibinfo{person}{J-P Martin} {and} \bibinfo{person}{Lorenzo
  Alvisi}.} \bibinfo{year}{2006}\natexlab{}.
\newblock \showarticletitle{Fast byzantine consensus}.
\newblock \bibinfo{journal}{\emph{IEEE Transactions on Dependable and Secure
  Computing}} \bibinfo{volume}{3}, \bibinfo{number}{3} (\bibinfo{year}{2006}),
  \bibinfo{pages}{202--215}.
\newblock


\bibitem[Pass and Shi(2018)]%
        {pass2018thunderella}
\bibfield{author}{\bibinfo{person}{Rafael Pass} {and} \bibinfo{person}{Elaine
  Shi}.} \bibinfo{year}{2018}\natexlab{}.
\newblock \showarticletitle{Thunderella: Blockchains with optimistic instant
  confirmation}. In \bibinfo{booktitle}{\emph{Advances in Cryptology--EUROCRYPT
  2018: 37th Annual International Conference on the Theory and Applications of
  Cryptographic Techniques, Tel Aviv, Israel, April 29-May 3, 2018 Proceedings,
  Part II 37}}. Springer, \bibinfo{pages}{3--33}.
\newblock


\bibitem[Shoup(2000)]%
        {shoup2000practical}
\bibfield{author}{\bibinfo{person}{Victor Shoup}.}
  \bibinfo{year}{2000}\natexlab{}.
\newblock \showarticletitle{Practical threshold signatures}. In
  \bibinfo{booktitle}{\emph{International conference on the theory and
  applications of cryptographic techniques}}. Springer,
  \bibinfo{pages}{207--220}.
\newblock


\bibitem[Shoup(2023)]%
        {shoup2023sing}
\bibfield{author}{\bibinfo{person}{Victor Shoup}.}
  \bibinfo{year}{2023}\natexlab{}.
\newblock \showarticletitle{Sing a song of Simplex}.
\newblock \bibinfo{journal}{\emph{Cryptology ePrint Archive}}
  (\bibinfo{year}{2023}).
\newblock


\bibitem[Shoup et~al\mbox{.}(2025)]%
        {shoup2025kudzu}
\bibfield{author}{\bibinfo{person}{Victor Shoup}, \bibinfo{person}{Jakub
  Sliwinski}, {and} \bibinfo{person}{Yann Vonlanthen}.}
  \bibinfo{year}{2025}\natexlab{}.
\newblock \showarticletitle{Kudzu: Fast and Simple High-Throughput BFT}.
\newblock \bibinfo{journal}{\emph{arXiv preprint arXiv:2505.08771}}
  (\bibinfo{year}{2025}).
\newblock


\bibitem[Shrestha et~al\mbox{.}(2025)]%
        {shrestha2025hydrangea}
\bibfield{author}{\bibinfo{person}{Nibesh Shrestha}, \bibinfo{person}{Aniket
  Kate}, {and} \bibinfo{person}{Kartik Nayak}.}
  \bibinfo{year}{2025}\natexlab{}.
\newblock \showarticletitle{Hydrangea: Optimistic Two-Round Partial Synchrony
  with One-Third Fault Resilience}.
\newblock \bibinfo{journal}{\emph{Cryptology ePrint Archive}}
  (\bibinfo{year}{2025}).
\newblock


\bibitem[Song and Van~Renesse(2008)]%
        {song2008bosco}
\bibfield{author}{\bibinfo{person}{Yee~Jiun Song} {and}
  \bibinfo{person}{Robbert Van~Renesse}.} \bibinfo{year}{2008}\natexlab{}.
\newblock \showarticletitle{Bosco: One-step byzantine asynchronous consensus}.
  In \bibinfo{booktitle}{\emph{International Symposium on Distributed
  Computing}}. Springer, \bibinfo{pages}{438--450}.
\newblock


\bibitem[Vonlanthen et~al\mbox{.}(2024)]%
        {vonlanthen2024banyan}
\bibfield{author}{\bibinfo{person}{Yann Vonlanthen}, \bibinfo{person}{Jakub
  Sliwinski}, \bibinfo{person}{Massimo Albarello}, {and} \bibinfo{person}{Roger
  Wattenhofer}.} \bibinfo{year}{2024}\natexlab{}.
\newblock \showarticletitle{Banyan: Fast rotating leader bft}. In
  \bibinfo{booktitle}{\emph{Proceedings of the 25th International Middleware
  Conference}}. \bibinfo{pages}{494--507}.
\newblock


\bibitem[Yin et~al\mbox{.}(2019)]%
        {yin2019hotstuff}
\bibfield{author}{\bibinfo{person}{Maofan Yin}, \bibinfo{person}{Dahlia
  Malkhi}, \bibinfo{person}{Michael~K Reiter}, \bibinfo{person}{Guy~Golan
  Gueta}, {and} \bibinfo{person}{Ittai Abraham}.}
  \bibinfo{year}{2019}\natexlab{}.
\newblock \showarticletitle{HotStuff: BFT consensus with linearity and
  responsiveness}. In \bibinfo{booktitle}{\emph{Proceedings of the 2019 ACM
  Symposium on Principles of Distributed Computing}}.
  \bibinfo{pages}{347--356}.
\newblock


\end{thebibliography}

\appendix

\section{Reproducing Experiments}\label{app:repro}

All experiments were executed with the \texttt{commonware-estimator}\footnote{\url{https://github.com/commonwarexyz/monorepo/tree/19f19d32760daf1d497295726ec92a1e6b84959f/examples/estimator}}. Each run combines one of the protocol scripts below with a network distribution. The workflow is:
\begin{enumerate}
    \item Save the desired script as \texttt{<protocol>.lazy}.
    \item Invoke \texttt{commonware-estimator --distribution <distribution> <protocol>.lazy} using one of the configurations listed at the end of this section.
\end{enumerate}

\subsection*{Protocol scripts}

\noindent The following schedules apply to experiments without erasure coding. Set \texttt{<proposal\_bytes>} to \texttt{32768} for the 32~KB runs and to \texttt{1048576} for the 1~MB runs.

\begin{tcolorbox}[title=Simplex,colback=gray!5!white,colframe=black!75!black,boxrule=0.5pt]
\begin{verbatim}
# Simplex
propose{0, size=<proposal_bytes>}
wait{0, threshold=1}
broadcast{1, size=40}
wait{1, threshold=67%}
broadcast{2, size=40}
wait{2, threshold=67%}
\end{verbatim}
\end{tcolorbox}

\begin{tcolorbox}[title=Kudzu (no coding),colback=gray!5!white,colframe=black!75!black,boxrule=0.5pt]
\begin{verbatim}
# Kudzu (no coding)
propose{0, size=<proposal_bytes>}
wait{0, threshold=1}
broadcast{1, size=40}
wait{1, threshold=61%}
broadcast{2, size=40}
wait{1, threshold=81%} || wait{2, threshold=61%}
\end{verbatim}
\end{tcolorbox}

\begin{tcolorbox}[title=Minimmit,colback=gray!5!white,colframe=black!75!black,boxrule=0.5pt]
\begin{verbatim}
# Minimmit
propose{0, size=<proposal_bytes>}
wait{0, threshold=1}
broadcast{1, size=40}
wait{1, threshold=41%}
wait{1, threshold=81%}
\end{verbatim}
\end{tcolorbox}

\vspace{0.2cm}
\noindent Use the erasure-coded variants below for Section~\ref{exp}'s coded experiments (message sizes already include shard data). 

\begin{tcolorbox}[title=Simplex (erasure coded),colback=gray!5!white,colframe=black!75!black,boxrule=0.5pt]
\begin{verbatim}
# Simplex (erasure coded)
propose{0, size=61682}
wait{0, threshold=1}
broadcast{1, size=61722}
wait{1, threshold=67%}
broadcast{2, size=40}
wait{2, threshold=67%}
\end{verbatim}
\end{tcolorbox}

\begin{tcolorbox}[title=Kudzu (erasure coded),colback=gray!5!white,colframe=black!75!black,boxrule=0.5pt]
\begin{verbatim}
# Kudzu (erasure coded)
propose{0, size=55190}
wait{0, threshold=1}
broadcast{1, size=55230}
wait{1, threshold=61%}
broadcast{2, size=40}
wait{1, threshold=81%} || wait{2, threshold=61%}
\end{verbatim}
\end{tcolorbox}

\begin{tcolorbox}[title=Minimmit (erasure coded),colback=gray!5!white,colframe=black!75!black,boxrule=0.5pt]
\begin{verbatim}
# Minimmit (erasure coded)
propose{0, size=55190}
wait{0, threshold=1}
broadcast{1, size=55230}
wait{1, threshold=41%}
wait{1, threshold=81%}
\end{verbatim}
\end{tcolorbox}

\subsection*{Network distributions}

\noindent Pair the scripts with one of the following network distributions when calling \texttt{commonware-estimator}. The Uniform global distribution is reused for both block sizes; only \texttt{<proposal\_bytes>} changes.

\begin{tcolorbox}[title=Uniform global (1~Gbps links),colback=gray!5!white,colframe=black!75!black,boxrule=0.5pt]
\begin{verbatim}
commonware-estimator --distribution \
  us-west-1:5:125000000,us-east-1:5:125000000,\
  eu-west-1:5:125000000,ap-northeast-1:5:125000000,\
  eu-north-1:5:125000000,ap-south-1:5:125000000,\
  sa-east-1:5:125000000,eu-central-1:5:125000000,\
  ap-northeast-2:5:125000000,ap-southeast-2:5:125000000 \
  <protocol>.lazy
\end{verbatim}
\end{tcolorbox}

\begin{tcolorbox}[title=Region-centric (1~Gbps links),colback=gray!5!white,colframe=black!75!black,boxrule=0.5pt]
\begin{verbatim}
commonware-estimator --distribution \
  us-west-1:13:125000000,us-east-1:12:125000000,\
  eu-west-1:3:125000000,ap-northeast-1:4:125000000,\
  eu-north-1:3:125000000,ap-south-1:3:125000000,\
  sa-east-1:3:125000000,eu-central-1:3:125000000,\
  ap-northeast-2:3:125000000,ap-southeast-2:3:125000000 \
  <protocol>.lazy
\end{verbatim}
\end{tcolorbox}

\begin{tcolorbox}[title=Uniform global (10~Gbps links),colback=gray!5!white,colframe=black!75!black,boxrule=0.5pt]
\begin{verbatim}
commonware-estimator --distribution \
  us-west-1:5:1250000000,us-east-1:5:1250000000,\
  eu-west-1:5:1250000000,ap-northeast-1:5:1250000000,\
  eu-north-1:5:1250000000,ap-south-1:5:1250000000,\
  sa-east-1:5:1250000000,eu-central-1:5:1250000000,\
  ap-northeast-2:5:1250000000,ap-southeast-2:5:1250000000 \
  <protocol>.lazy
\end{verbatim}
\end{tcolorbox}

\section{Erasure codes: the setup} \label{erasuresetup} 

In this section, we describe the (standard) cryptographic primitives required for Appendices \ref{erasureintu}-\ref{erasureanalysis}. 

\vspace{0.2cm}
\noindent \textbf{Threshold signatures}. A $k$-of-$n$ threshold signature scheme allows \emph{signature shares} from any $k$ processors on a given message to be combined to form a \emph{certificate} on that message. Forming a certificate on any message is infeasible given less than $k$ signature shares.  Such schemes can be implemented using BLS signatures \cite{boneh2001short}.  We use one threshold signature scheme with $k=2f+1$, and another with $k=n-f$. 

\vspace{0.2cm}
\noindent \textbf{Erasure codes}.
 We suppose given an $(n,2f+1)$-erasure code, which uniquely encodes any bit string $C$ of length $\beta$ as a sequence of $n$ \emph{fragments}, $c_1,\dots,c_n$,  in such a way that any $2f+1$  fragments and $\beta$ suffice to efficiently reconstruct $C$. We suppose all fragments have the same size (as a function of $n,f$ and $\beta$). 
Reed-Solomon codes can be used to realise an $(n,2f+1)$-erasure code so that each fragment has size $\approx \beta/(2f+1)$. If $n=5f+1$, this leads to a data expansion rate of roughly $2.5$, i.e., the combined size of all $n$ fragments is roughly 2.5$\beta$. 

\vspace{0.2cm}
\noindent \textbf{Merkle trees}. We use Merkle trees in the standard way to allow a processor $p$ to commit to a sequence of values $v_1,\dots,v_n$. To form the commitment, $p$ constructs a full binary tree in which the leaves are the hashes of $v_1,\dots,v_n$ and every other node is the hash of its two children. The commitment is the root of the tree, $r$ say. 
To \emph{open} the commitment at position $i$, $p$ specifies $v_i$ and a \emph{validation path  from $r$ to $v_i$ at position $i$}: the validation path specifies the sibling of each node on the path from the hash of $v_i$ at position $i$ to the root. 

\vspace{0.2cm}
\noindent \textbf{Encoding, certified fragments, and tags}. We use techniques introduced by Cachin and Tessaro \cite{cachin2005asynchronous} for the purpose of \emph{asynchronous verifiable information dispersal (AVID)}. These techniques combine the use of Merkle trees and erasure codes. Given a bit string $C$ of length $\beta$, let $c_1,\dots,c_n$ be the corresponding fragments produced by our $(n,2f+1)$-erasure code. Form a Merkle tree whose leaves are the hashes of $c_1,\dots, c_n$ and let $r$ be the root of this tree. For each $i\in [n]$, let $\pi_i$ be a validation path  from $r$ to $v_i$ at position $i$. We define the \emph{tag} $\tau(C):= (\beta, r)$ and set: 
\[ \text{Encode}(C):= (\tau(C), \{ (c_i,\pi_i ) \}_{i\in [n]} ). \]

\noindent If $z= (\beta', r')$ for some $\beta'\in \mathbb{N}$ and some hash value $r'$, we say $(c,\pi)$ a \emph{certified fragment of $z$ at $i$} if both: 
\begin{itemize} 
\item $c$ is of the correct length (given $n$ and $f$) to be a fragment of a message of length $\beta'$, and;
\item $\pi$ is a validation path from $r'$ to $c$ at position $i$. 
\end{itemize}

\vspace{0.2cm}
\noindent \textbf{Decoding}. The function Decode takes inputs of the form
\[ (z,  \{ (c_i,\pi_i ) \}_{i\in I} ), \]
where $z=(\beta, r)$ for $\beta\in \mathbb{N}$, $r$ is a hash value, $I\subseteq [n]$ with $|I|=2f+1$, and each $(c_i,\pi_i)$ is a certified fragment of $z$ at $i$. It then reconstructs a message $C$ of length $\beta$ from the fragments $\{ c_i \}_{i\in I}$. If this reconstruction fails (e.g., due to a formatting error) it outputs $\bot$. Otherwise, it computes $\tau(C)=(\beta, r')$. If $r\neq r'$ it outputs $\bot$, and otherwise outputs $C$.

\section{Erasure coding for Minimmit: the intuition} \label{erasureintu} 

In this section, we discuss the intuition behind the modifications required to describe a version of Minimmit, called E-Minimmit, which efficiently integrates the use of erasure codes.

To explain our approach, we first consider a simple (but non-optimal) version, which has a data-expansion rate of 5, i.e., the total size of all fragments sent out by each leader to all processors combined is roughly five times the size of a block. The factor of 5 arises because this simple protocol requires that $f+1$ fragments suffice to reconstruct a block. Then we describe a modified approach using an $(n,2f+1)$ erasure code, resulting in a data-expansion rate of 2.5. 

\subsection{A simple (but non-optimal) approach} \label{simple}

If one is happy to accept a data-expansion rate of 5, then erasure codes can be integrated with Minimmit in an entirely standard way, using methods first described in \cite{cachin2005asynchronous}. Since E-Minimmit will use threshold signatures, we describe an approach that does the same:  
\begin{itemize}
    \item  Rather than the leader $p$ of view $v$ sending the entire block $b$ to each processor, it first separates out the \emph{payload} $C:=b.\text{Tr}$ of the block, i.e., the sequence of transactions associated with the block.
    \item It then calculates Encode$(C)=(\tau(C), \{ (c_i,\pi_i ) \}_{i\in [n]} )$ (as described in Appendix \ref{erasuresetup}, but now using an $(n,f+1)$-erasure code).
    \item The block $b$ itself is now just a tuple $(v,\tau(C),h)$ signed by $p$, where $h$ is the hash of the parent of $b$. We refer to the message $(b,i,c_i,\pi_i)$ as \emph{the certified fragment of $b$ at $i$}. 
    \item The leader $p$ then sends each $p_i$ the certified fragment of $b$ at $i$. 
    \item Upon checking that this message is well-formed, including verifying 
   $\pi_i$ and that the necessary notarisations and nullifications (or the corresponding threshold certificates) have been received, $p_i$ then disseminates the certified fragment of $b$ at $i$ as well as a vote for $b$. The vote for $b$ contains  $p_i$'s signature share for both the M-notarisation and 
   L-notarisation threshold schemes on the message $(\text{vote}, b)$.
    \item The remaining instructions are largely unchanged, except that,  instead of requiring processors to forward on M-notarisations, we have processors form a certificate (an `M-certificate') from the signature shares included in votes, using a $(2f+1)$-of-$n$ threshold signature scheme. Processors now forward these certificates instead of a set of $2f+1$ votes (and we proceed similarly for nullifications).
    \item The factor 5 arises because the existence of an M-certificate only suffices to ensure that $(2f+1)-f=f+1$ correct processors have sent their fragments to all processors. 
    So, we require that $f+1$ fragments suffice for block payload reconstruction. 
\end{itemize} 

Next, we describe the intuition behind reducing the data-expansion rate to 2.5. 

\subsection{Reducing the data-expansion rate to 2.5} \label{2.5intu}

To reduce the data-expansion rate, we use an $(n,2f+1)$-erasure code. This means the existence of an M-certificate no longer ensures sufficiently many correct processors have disseminated their fragments to guarantee data-availability for the block payload. Accordingly, we must now allow transition from view $v$ upon receiving either:
\begin{itemize}
\item[(a)] A nullification (or the corresponding certificate) for the view, or;
\item[(b)] An M-certificate for some view $v$ block \emph{and}  fragments from sufficiently many processors  to reconstruct the block. 
\end{itemize} 
With this change in place, there are two further required changes to the protocol: 

\vspace{0.2cm} 
\noindent \textbf{The sending of others' fragments (in the bad case).} Suppose that $p_i$ is in view $v$ and finds that (b) above applies with respect to some block $b$. Then $p_i$ should immediately progress to the next view. However, since $p_i$ may be the leader of some subsequent view (and may propose a block with parent $b$ that it wishes others to vote for), it must also guarantee that other correct processors will be able to reconstruct the block payload. Since up to $f$ of the processors that sent fragments to $p_i$ may be faulty, this is not necessarily the case. To deal with this issue,  we first consider an approach that leads to large time-outs. Then we consider how to reduce time-outs. 

\vspace{0.2cm} 
\noindent \emph{An approach that gives time-outs of $8\Delta$}. One approach $p_i$ could follow would be to select a set of $f$ processors disjoint from those whose fragments $p_i$ has received, and to send those processors notification that $p_i$ can send them their fragment if they have not already received it. This ensures that $3f+1$ processors have access to their own fragment, so that $2f+1$ correct processors have access to their own fragment. Data-availability is ensured so long as those correct processors disseminate their fragment to others. 

While this approach would work, it has the downside of significantly increasing the required time-out for each view. Suppose the first correct processor $p_i$ to enter view $v$ does so at $t\geq \text{GST}$ because they find that (b) above applies w.r.t.\ some block $b_1$. Suppose $\mathtt{lead}(v)$ is correct and that we wish to ensure they propose a new block that becomes finalised. Then all correct processors, including $\mathtt{lead}(v)$, will enter the view by $t+4\Delta$: we must wait one $\Delta$ for others to receive $p_i$'s messages (indicating that it can send required fragments), another $\Delta$ for $p_i$ to receive responses to those messages, one $\Delta$ for $p_i$ to send the relevant fragments, and then a final $\Delta$ for those fragments to be disseminated. However, $\mathtt{lead}(v)$ may enter view $v$ just before it receives sufficient fragments to recover $b_1$ because (b) above applies w.r.t. a different block $b_2$, and may propose a block $b_3$ with $b_2$ as parent. Using the same argument, we can only guarantee that correct processors will be able to recover $b_2$ and vote for $b_3$ by $t+8\Delta$. This would mean time-outs of $8\Delta$, meaning that crashed leaders cause significant delays. 

\vspace{0.2cm} 
\noindent \emph{An approach with reduced time-outs}. We follow an approach that gives significantly reduced time-outs, while requiring no extra message sending in the case that leaders are correct during synchrony. The approach uses a parameter $s$, which should be set to roughly the expected time between receiving $2f+1$ votes and receiving $3f+1$ votes when the leader is correct during synchrony. Generally, this will be some small fraction of $\Delta$. Upon leaving view $v$ because it has received an M-certificate for $b$ and successfully decoded the payload, $p_i$ now sets a timer to expire in time $s$. The intuition is that if the leader was correct during synchrony, 
$p_i$ will receive $3f+1$ votes within time $s$, ensuring $2f+1$ correct 
processors have their fragments—sufficient for data availability. Only if 
this threshold is not met does $p_i$ need to intervene. Once the timer expires, processor $p_i$: 
\begin{itemize} 
\item Sets $I:= \{ j\in [n]:\ p_i \text{ has received a certified fragment of }b\text{ at }j \text{ from }p_j \}$; 
\item Sets $x:=(3f+1)-|I|$, and if $x>0$ then it:
\begin{enumerate} 
\item[(i)] Chooses  $J\subseteq [n]$ with $|J|=x$ and $J\cap I=\emptyset$; 
\item[(ii)] For each $j\in J$, sends $p_j$ a certified fragment of $b$ at $j$.
\end{enumerate} 
\end{itemize} 
The consequence is that by time $s+\Delta$ after $p_i$ leaves view $v$ at time $t$, 
at least $3f+1$ processors possess their fragments (if necessary, by $p_i$'s intervention), guaranteeing that $2f+1$ correct processors 
have disseminated their fragments. This ensures all correct processors leave view $v$ by $t+2\Delta+s$. The same argument as above shows that timeouts can now be set to $4\Delta +2s$.  

\vspace{0.2cm} 
\noindent \emph{An optimisation reducing timeouts to $3\Delta$}. We will initially formalise a version of E-Minimmit that ensures every correct leader after GST finalises a new block, and which uses timeouts of $4\Delta +2s$. However, in Appendix \ref{erasureanalysis}, we will also describe an optimisation that uses timeouts of $3\Delta$. The trade-off is that view $v$ is then only guaranteed to produce a new finalised block if view $v-1$  begins after GST and the leaders of both views $v-1$ and $v$ are correct.

\vspace{0.2cm} 
\noindent \textbf{Late dissemination of one's own fragment (in the bad case).} 
The protocol described above ensures data availability through two mechanisms: 
processors disseminate fragments with their  votes, and processors who 
leave views early may subsequently send fragments to others. However, a subtle 
issue remains. Consider a processor $p_j$ that receives a certified fragment of block $b$ from 
some processor $p_i$ (via the mechanism described above) after $p_j$ has already 
voted for a different block in view $v$ or after $p_j$ has timed out. In this 
case, $p_j$ cannot produce a vote for $b$, and therefore would not automatically 
disseminate its fragment for $b$.

To ensure data availability in such scenarios, we require that processors 
disseminate their own fragments 
whenever they hold a fragment for a block with an M-certificate but have not 
yet disseminated that fragment. 
This ensures that $2f+1$ correct processors eventually disseminate their 
fragment, guaranteeing data availability.\footnote{Note that in executions without Byzantine behaviour, each processor still 
disseminates its fragment for at most one block per view. Byzantine leaders can produce slashable behaviour causing some correct processors to disseminate two extra fragments.  This behaviour is `slashable', because it requires the leader to produce more than one block for the same view, leading to objective proof of Byzantine action (for a formal account of `accountability/slashability' see \cite{lewis2025beyond}).}

\subsection{Avoiding the finalisation of junk blocks} \label{junkintu} 
We introduce one further modification to address an issue arising from erasure coding: since processors vote for blocks before decoding their payloads, Byzantine leaders can cause blocks with invalid payloads to receive L-notarisations. Whilst one could simply finalise such blocks and subsequently ignore their contents when forming the transaction sequence, we prefer to exclude them entirely from the finalised chain. We achieve this through a straightforward mechanism.

The key insight is to decouple M-certificates from block inclusion in the finalised chain. We proceed as follows: 
\begin{enumerate}
    \item Each processor $p_i$ maintains a local variable $\mathtt{blocks}$ containing the genesis block and all blocks $b$ satisfying: 
    \begin{itemize} 
    \item[(i)] $p_i$ has received (or formed) an M-certificate for $b$;
    \item[(ii)] $p_i$ has received sufficient fragments to decode the payload and verify it is well-formed, and;
    \item[(iii)] The parent of $b$ belongs to $\mathtt{blocks}$. 
    \end{itemize} 
    \item Processors advance from view $v$ upon receiving either a nullification (or the corresponding certificate) for view $v$, or upon finding that $\mathtt{blocks}$ contains a view $v$ block. 
    \item Processors finalise only those blocks appearing in (their local copy of) $\mathtt{blocks}$ (upon receipt of L-notarisations or L-certificates). 
\end{enumerate}

This approach introduces a liveness concern: if view $v$ produces a block with an M-certificate but invalid payload, we must ensure processors can still progress beyond view $v$. Since the block never enters $\mathtt{blocks}$ (failing condition (ii)), processors may require a nullification for view $v$ to advance.

We address this by requiring processors to issue nullify$(v)$ messages upon decoding a view $v$ block with an M-certificate and discovering the payload is invalid. This preserves consistency: the standard quorum intersection argument establishes that if view $v$ produces a block with a valid payload that achieves finalization, then no alternative view $v$ block can receive an M-certificate. Consequently, no processor will decode a distinct view $v$ block, discover it has an invalid payload, and issue a nullify$(v)$ message. This ensures that nullifications for view $v$ cannot reach the $2f+1$ threshold required.

\section{E-Minimmit: the formal specification} \label{erasurespec} 

As in Section \ref{formal},  we say `disseminate' to mean `send to all processors'. When a correct processor is instructed to send a message to itself, it regards that message as immediately
received.   The pseudocode uses a number of message types, local
variables, functions and procedures, detailed below. While some  of the following definitions are the same as in Section \ref{formal}, we repeat them here for ease of reference. 

\vspace{0.2cm}
\noindent \textbf{The function} $\mathtt{lead}(v)$.  As in Section \ref{formal},  we set $\mathtt{lead}(v):= p_{j+1}$, where $j= v \text{ mod }n$.

\vspace{0.2cm}
\noindent \textbf{Blocks}. Recall that $\tau$ is the tag function, as defined in Appendix \ref{erasuresetup}. The \emph{genesis block} is the tuple $b_{\text{gen}}:=(0,\tau(\lambda),\lambda)$, where $\lambda$ denotes the empty sequence (of length 0).  A block other than the genesis block, with associated payload $C$,  is a tuple $b=(v,\tau(C), h)$ signed by $\mathtt{lead}(v)$, where:
\begin{itemize}
\item $v\in \mathbb{N}_{\geq 1}$ is the view corresponding to $b$;
\item $\tau(C)$ is the tag resulting from $C$;
\item $h$ is the hash of $b$'s parent block.
\end{itemize}
We write $b.\text{view}$, $b.\text{tag}$ and $b.\text{par}$ to denote the corresponding entries of $b$. If $b.\text{view}=v$, we also refer to $b$ as a `view $v$ block'. If $(c_i,\pi_i)$ is a certified fragment of $\tau(C)$ at $i$, we also say that the tuple $(b,i,c_i,\pi_i)$ is \emph{certified fragment of} $b$ at $i$: it is a message of this form that a correct leader of view $v$ will send to $p_i$ while in view $v$. 

\vspace{0.2cm}
\noindent \textbf{Votes}. A \emph{vote} by $p_i\in \Pi$ for the block $b$ is a message of the form $(\text{vote},b,i,\rho_i,\rho_i')$, where: 
\begin{itemize} 
\item $\rho_i$ is a (valid) signature share from $p_i$ on the message $(\text{vote},b)$, using a $(2f+1)$-of-$n$ threshold signature scheme; 
\item $\rho_i'$ is a signature share from $p_i$ on the message $(\text{vote},b)$, using an $(n-f)$-of-$n$ threshold signature scheme. 
\end{itemize} 
Processor $p_i$ may disseminate a vote for $b$ upon receiving a certified fragment of $b$ at $i$.  The vote for $b$ by $p_i$ contains signature shares from $p_i$ corresponding to $(2f+1)$-of-$n$ and $(n-f)$-of-$n$ threshold signature schemes. As detailed below, the first of these signature shares can be used to help form an `M-certificate' for $b$, while the second can be used to form an `L-certificate' for $b$.

\vspace{0.2cm}
\noindent \textbf{M-notarisations and certificates}. An M-\emph{notarisation} for the block $b$ is a set of $2f+1$ votes for $b$, each by a different processor in $\Pi$. (By an M-\emph{notarisation}, we mean an M-notarisation for some block.) 
An M-\emph{certificate} for the block $b$ is the message $(\text{M-cert},b,\rho)$, where $\rho$ is a $(2f+1)$-of-$n$ threshold certificate on the message $(\text{vote},b)$. 

\vspace{0.2cm}
\noindent \textbf{L-notarisations and certificates}. An \emph{L-notarisation} for the block $b$ is a set of $n-f$ votes for $b$, each by a different processor in $\Pi$. 
 An L-\emph{certificate} for the block $b$ is the message $(\text{L-cert},b,\rho)$, where $\rho$ is an $(n-f)$-of-$n$  certificate on the message $(\text{vote},b)$.

\vspace{0.2cm}
\noindent \textbf{Nullify messages and nullifications}. For $v\in \mathbb{N}_{\geq 1}$, a nullify$(v)$ message by $p_i$ is of the form $(\text{nullify},v,i,\rho_i)$, where $\rho_i$ is a (valid) signature share from $p_i$ on the message $(\text{nullify},v)$, using a $(2f+1)$-of-$n$ threshold signature scheme. A \emph{nullification} for view $v$ is a set of $2f+1$ nullify$(v)$ messages, each by a different processor in $\Pi$. (By a \emph{nullification}, we mean a nullification for some view.) An \emph{N-certificate} for view $v$ is a message $(\text{N-cert}, v,\rho)$, where $\rho$ is a $(2f+1)$-of-$n$ certificate on the message $(\text{nullify,}v)$.

\vspace{0.2cm}
\noindent \textbf{The local variable} $\mathtt{v}$. Initially set to 1, this variable specifies the present view of a processor.

\vspace{0.2cm}
\noindent \textbf{The local variable} $\mathtt{b}$. Initially set to $b_{\text{gen}}$, this variable is used by leaders to choose a parent block. 

\vspace{0.2cm}
\noindent \textbf{The local timer} $\mathtt{T}$. Each processor $p_i$ maintains a local timer $\mathtt{T}$, which is initially set to 0 and increments in real-time. (Processors will be explicitly instructed to reset their timer to 0 upon entering a new view.)

\vspace{0.2cm}
\noindent \textbf{The local variable} $\mathtt{certificates}$. Initially empty, this local variable for $p_i$ is automatically updated by $p_i$ (without explicit instructions in the pseudocode) to contain all M-certificates, L-certificates, and N-certificates received by $p_i$. If $p_i$ receives an M-notarisation, L-notarisation, or a nullification, then it automatically forms the associated certificate, and adds it to $\mathtt{certificates}$. 

An N-certificate $Q$ for view $v$ is regarded as \emph{new} at timeslot $t$ if $Q\in \mathtt{certificates}$ and $ \mathtt{certificates}$ did not contain an N-certificate for view $v$ at any timeslot $t'<t$. Similarly, an M-certificate $Q$ for $b$ is new at timeslot $t$ if $Q\in \mathtt{certificates}$ and $ \mathtt{certificates}$ did not contain an M-certificate for $b$ at any timeslot $t'<t$.

\vspace{0.2cm}
\noindent \textbf{The local variable} $\mathtt{blocks}$. Initially set to $\{ b_{\text{gen}} \}$, this local variable is automatically updated by $p_i$ (without explicit instructions in the pseudocode). At any timeslot, $\mathtt{blocks}$ contains $b_{\text{gen}}$ and all blocks $b$ such that (i)-(iii) below are all satisfied: 
\begin{itemize}
\item[(i)] $\mathtt{certificates}$ contains an  M-certificate for $b$;
\item[(ii)] There exists $I\subset [n]$ with $|I|=2f+1$ such that, for each $j\in I$, $p_i$ has received $(b,j,c_j,\pi_j)$ from $p_j$, which is a certified fragment of $b$ at $j$.  On input $(b.\text{tag}, \{ (c_j,\pi_j) \}_{j\in I})$, Decode does not output $\bot$;
\item[(iii)] There exists $b'\in \mathtt{blocks}$ with $H(b')=b.\text{par}$, i.e., the parent of $b$ is in $\mathtt{blocks}$. 
\end{itemize} 

\noindent At timeslot $t$, $b\in \mathtt{blocks}$ is regarded as \emph{new} if it was not in $\mathtt{blocks}$ at any previous timeslot. 

\vspace{0.2cm}
\noindent \textbf{The local timers} $\mathtt{T}(b)$. Processor $p_i$ will leave view $v$ upon  finding that $\mathtt{blocks}$ contains a view $v$ block $b$, or else upon finding that $\mathtt{certificates}$ contains an N-certificate for view $v$. In the former case, it will not be necessary during `standard operation' for $p_i$ to send out further certified fragments of $b$. However, if the leader of the view is faulty, it may be necessary to do so, so  that other processors can recover the payload associated with $b$. 

To achieve this, we consider a parameter $s$, which should be set to roughly the expected time between receiving $2f+1$ votes for $b$ and receiving $3f+1$ votes for $b$,  during synchrony and if the leader of the view is correct.  Upon putting a new block $b$ into $\mathtt{blocks}$,  $p_i$ \emph{triggers} a timer $\mathtt{T}(b)$ to expire in time $s$. When the timer expires, the pseudocode specifies to which processors $p_i$ must send their corresponding certified fragments of $b$.  

\vspace{0.2cm}
\noindent \textbf{The local variables} $\mathtt{nullified}$, $\mathtt{proposed}$, and $\mathtt{notarised}$. These are used by $p_i$ to record whether it has yet sent a nullify$(\mathtt{v})$ message, whether it has yet proposed a block for view $v$, and the block  it has voted for in the present view: $\mathtt{nullified}$ and $\mathtt{proposed}$ are initially set to false, while $\mathtt{notarised}$ is initially set to $\bot$ (a default value different than any block). These values will be explicitly reset upon entering a new view.

\vspace{0.2cm}
\noindent \textbf{The procedure} ProposeBlock. This procedure is executed by the leader $p_i$ of view $v$ to produce and send out a new block. To execute the procedure, $p_i$:
\begin{itemize}
\item  Forms a payload $C$, containing all received transactions not included the payloads of ancestors of $\mathtt{b}$.\footnote{The protocol ensures $\mathtt{b}\in \mathtt{blocks}$, so that $p_i$ knows the payload of every ancestor of $\mathtt{b}$.} 
\item Calculates Encode$(C):=(\tau(C), \{ (c_j,\pi_j ) \}_{j\in [n]} )$ and sets $b:=(v,\tau(C),H(\mathtt{b}))$ signed by $p_i$.  
\item For each $j\in [n]$, sends $(b,j,c_j,\pi_j)$ to $p_j$, i.e., sends the certified fragment of $b$ at  $j$ to $p_j$. 
\end{itemize}

 \vspace{0.2cm}
\noindent \textbf{When $p_i$ has received a votable fragment for view $v$}. If $b=(v,\sigma,h)$ is a view $v$ block signed by $\mathtt{lead}(v)$ and $(b,i,c_i,\tau_i)$ is a certified fragment  of $b$ at $i$, then $p_i$ regards this certified fragment as \emph{votable for view $v$} if: 
\begin{enumerate}
\item[(i)] there exists $b'\in \mathtt{blocks}$ with $H(b')=h$, $b'.\text{view}<v$, and; 
\item[(ii)] $\mathtt{certificates}$ contains an N-certificate for each view in the open interval $(b'.\text{view},v)$.
\end{enumerate}

\vspace{0.2cm} For ease of reference, local variables are displayed in the table below. The pseudocode appears in Algorithm 2.

\begin{table}[h!]
  \begin{center}
    \label{tab:table8}
    \begin{tabular}{l|l} 
      \textbf{Variable} & \textbf{Description} \\
      \hline
      $\mathtt{v}$ & Initially 1, specifies the present view \\
       $\mathtt{b}$ & Initially $b_{\text{gen}}$, used to specify parents \\
      $\mathtt{T}$ & Initially 0, a local timer reset upon entering each view \\
        $\mathtt{T}(b)$ & Expires after time $s$ once triggered\\
$\mathtt{nullified}$ & Initially false, specifies whether already sent nullify$(\mathtt{v})$ message \\
$\mathtt{proposed}$ & Initially false, specifies whether already proposed a block for view $\mathtt{v}$ \\
$\mathtt{notarised}$ & Initially set to $\bot$, records block voted for in present view \\
$\mathtt{certificates}$ & Initially empty, records all received certificates \\
 & Automatically updated \\
 $\mathtt{blocks}$ & Initially contains only $b_{\text{gen}}$, records all blocks with M-certificates \\
 &  and recovered payloads. Automatically updated \\

    \end{tabular}
        \caption{Local variables for E-Minimmit}

  \end{center}
\end{table}

 \begin{algorithm} \label{alg2}
\caption{: the instructions for $p_i$ at timeslot $t$}
\begin{algorithmic}[1]

\scriptsize

\State  At every timeslot $t$:   \Comment MAIN LOOP 

\State 

   \State  \hspace{0.1cm} Disseminate new N-certificates;  \label{Ndis}  \Comment `new' as defined in Section \ref{erasurespec}

      \State  \hspace{0.1cm} Disseminate new M-certificates;  \label{Mdis} 

\State 

 \State   \hspace{0.1cm} If $p_i=\mathtt{lead}(\mathtt{v})$ and $\mathtt{proposed}=$ false:

     \State \hspace{0.3cm} $\text{ProposeBlock}$; Set $\mathtt{proposed}:=$ true; \label{esendblock}   \Comment Send out a new block

     \State

      \State   \hspace{0.1cm} If $p_i$ has received $(b,i,c_i,\pi_i)$ that is votable for view $\mathtt{v}$:  \Comment `votable' as defined in Section \ref{formal}
      \State \hspace{0.3cm}  If $\mathtt{notarised}=\bot$ and $\mathtt{nullified}=$ false:  \label{evotecheck}
      \State \hspace{0.6cm} Set $\mathtt{notarised}:=b$ and disseminate a vote for $b$ by $p_i$; \label{evote1} \Comment Send vote
      \State \hspace{0.6cm} Disseminate  $(b,i,c_i,\pi_i)$ if $p_i$ has not already done so;

      \State
        \State   \hspace{0.1cm} If $\mathtt{T}=4\Delta+2s$, $\mathtt{nullified}=$ false and $\mathtt{notarised}=\bot$:  \label{etimeout1}
        \State \hspace{0.3cm} Set $\mathtt{nullified}:=$ true and disseminate a $\text{nullify}(\mathtt{v})$ message by $p_i$; \label{etimeout2} \Comment Send nullify$(\mathtt{v})$ upon time-out

        \State

        \State   \hspace{0.1cm}  If  $\mathtt{certificates}$ contains an N-certificate for $\mathtt{v}$:
   
        \State  \hspace{0.3cm}  Set $\mathtt{v}:=\mathtt{v}+1$, $\mathtt{nullified}:=$ false, $\mathtt{proposed}:=$ false, $\mathtt{notarised}:=\bot$, $\mathtt{T}:=0$; \label{enewview1} 
         \Comment Go to next view
         \State   \hspace{0.1cm}  If $\mathtt{blocks}$ contains a view $\mathtt{v}$ block $b$, 
          \State  \hspace{0.3cm} If  $\mathtt{notarised} =\bot$ and $\mathtt{nullified}=$ false, disseminate a vote for $b$ by $p_i$;  \label{evote2} \Comment Send vote (akin to line \ref{vote2} of Algorithm 1)

         \State  \hspace{0.3cm}  Set  $\mathtt{b}:=b$, $\mathtt{v}:=\mathtt{v}+1$, $\mathtt{nullified}:=$ false, $\mathtt{proposed}:=$ false, $\mathtt{notarised}:=\bot$, $\mathtt{T}:=0$; \label{enewview2} 
         \Comment Update $\mathtt{b}$ and go to next view

    \State 
              \State   \hspace{0.1cm}  If $\mathtt{certificates}$ contains an L-certificate for any non-finalised $b\in \mathtt{blocks}$: 
               \State \hspace{0.3cm} Finalise $b$ and all ancestors;   \Comment Finalisation

    \algrule
    
    \State \Comment EXTRA INSTRUCTIONS TO SEND NULLIFY MESSAGES

      \State  \hspace{0.1cm} If $\mathtt{nullified}=$ false, $\mathtt{notarised}\neq \bot$ and there exists $I\subset [n]$ with $|I|=2f+1$ and a set   \label{ebeginN}
      \State \hspace{0.1cm}of messages $\{ m_j \}_{j\in I} $ s.t.,  for each $j\in I$, $m_j$ has been received by $p_i$ and is either: 
      \State  \hspace{0.1cm} (i) a \text{nullify}$(\mathtt{v})$ message by $p_j$,  or;
      \State  \hspace{0.1cm} (ii) a vote by $p_j$ for some view $\mathtt{v}$ block  $b\neq \mathtt{notarised}$:   \label{eendN}

         \State  \hspace{0.3cm} Set $\mathtt{nullified}:=$ true;
         Disseminate a  $\text{nullify}(\mathtt{v}) $ message by $p_i$; \label{esendN}
            \State \Comment Disseminate nullify$(\mathtt{v})$ message upon proof of no progress for $\mathtt{v}$ (akin to lines \ref{beginN}-\ref{sendN} of Algorithm 1)

          \State

      \State \hspace{0.1cm} If  $\mathtt{nullified}=$false and there exists a view $\mathtt{v}$ block $b$ such that:  \label{n3a} 
      
      \State  \hspace{0.1cm} (i) $\mathtt{certificates}$ contains an  M-certificate for $b$.
      
         \State  \hspace{0.1cm} (ii) There exists $I\subset [n]$ with $|I|=2f+1$ such that, for each $j\in I$, $p_i$ has received $(b,j,c_j,\pi_j)$ from $p_j$, which is a certified
         \State  \hspace{0.1cm}   fragment of $b$ at $j$.  On input $(b.\text{tag}, \{ (c_j,\pi_j) \}_{j\in I})$, Decode outputs $\bot$.  \label{n3b}

               \State  \hspace{0.3cm} Set $\mathtt{nullified}:=$ true; 
               Disseminate a  $\text{nullify}(\mathtt{v})$ message by $p_i$; \label{esendN2}
                \State \Comment Disseminate nullify$(\mathtt{v})$ message upon notarisation for junk block (see Section \ref{junkintu})

        \algrule
          \State \Comment INSTRUCTIONS TO SEND EXTRA FRAGMENTS IN THE BAD CASE (see Section \ref{2.5intu})
          
          \State \hspace{0.1cm} For each new $b\in \mathtt{blocks}$: 
          
          \State \hspace{0.3cm} Trigger the timer $\mathtt{T}(b)$ to expire in time $s$;  \Comment $s$ a parameter
          
           \State

          \State \hspace{0.1cm}  For each $b$ such that $\mathtt{T}(b)$ expires at $t$:   \label{f1a} 

          \State \hspace{0.3cm} Set $I:= \{ j\in [n]:\ p_i \text{ has received a certified fragment of }b\text{ at }j \text{ from }p_j \}$; 
          \State \hspace{0.3cm} Set $x:=(3f+1)-|I|$, and if $x>0$:
     
          \State \hspace{0.6cm} Choose $J\subseteq [n]$ with $|J|=x$ and $J\cap I=\emptyset$; 
          \State \hspace{0.6cm} For each $j\in J$, send $p_j$ a certified fragment of $b$ at $j$;  \label{f1b} \Comment Send others their fragment

\State 

\State  \hspace{0.1cm}  For any block $b$ s.t.:  \label{f2a} 
\State   \hspace{0.1cm}  (i) $\mathtt{certificates}$ contains an M-certificate for $b$; 
\State   \hspace{0.1cm}  (ii) $p_i$ has received a certified fragment of $b$ at $i$, and;
\State  \hspace{0.1cm}  (iii) $p_i$ has not previously disseminated a certified fragment of $b$ at $i$: 
\State \hspace{0.3cm} Disseminate a certified fragment of $b$ at $i$;     \Comment Send own fragment  \label{f2b}

\end{algorithmic}
\end{algorithm}

\section{E-Minimmit: analysis} \label{erasureanalysis} 

\subsection{Consistency}  \label{econsec}
 We say a block $b$ receives an M-certificate if some processor receives an M-certificate for $b$ (and similarly for L-certificates). View $v$ receives an N-certificate if some processor receives an N-certificate for view $v$. The proof of consistency is very similar to that for Minimmit. The only substantive change is in the proof is the corresponding version of  Lemma \ref{X2}.

\begin{lemma}[One vote per view] \label{esinglevote}
Correct processors vote for at most one block in each view. 
\end{lemma}
\begin{proof}
The proof is the same as for Lemma \ref{singlevote}. 
\end{proof}

\begin{lemma}[(X1) is satisfied] \label{eX1}  If  $b$ receives an  L-certificate, then no block $b'\neq b$ with $b'.\text{view}=b.\text{view}$ receives an M-certificate. 
\end{lemma}
\begin{proof}
The proof is essentially identical to that for Lemma \ref{X1}. 
\end{proof}

\begin{lemma}[(X2) is satisfied]  \label{eX2} If  $b$ receives an L-certificate and some  correct processor finalises $b$,   then view $v:=b.\text{view}$ does not receive an N-certificate.
\end{lemma}
\begin{proof}
Towards a contradiction, suppose $b$ receives an L-certificate and some  correct processor finalises $b$, let $v:=b.\text{view}$,  and suppose view $v$ receives an N-certificate. Let $P$ be the correct processors that vote for $b$, let $P'=\Pi\setminus P$, and note that $|P'|\leq 2f$. Since view $v$ receives an N-certificate, it follows that some processor in $P$ must send a nullify$(v)$ message. So, let $t$ be the first timeslot at which some processor $p_i\in P$ sends such a message. Since $p_i$ cannot send a nullify$(v)$ message upon timeout (lines \ref{etimeout1}-\ref{etimeout2}), there are two possibilities. The first is that $p_i$ sends the nullify$(v)$ message at $t$ because the conditions of lines  \ref{ebeginN}-\ref{eendN} hold for $p_i$ at $t$, i.e., there exists  $I\subset [n]$ with $|I|=2f+1$ and a set of messages $\{ m_j \}_{j\in I} $ s.t.,  for each $j\in I$, $m_j$ has been received by $p_i$ and is either: 
\begin{enumerate} 
\item[(i)]  A \text{nullify}$(v)$ message by $p_j$,  or;
\item[(ii)] A vote by $p_j$ for some view $v$ block  $b'\neq b$. 
\end{enumerate} 
By Lemma \ref{esinglevote}, no processor in $P$ sends a message of form (ii). By our choice of $t$, no processor in $P$ sends  a message of form (i) prior to $t$. Combined with the fact that $|P'|\leq 2f$, this gives the required contradiction.

The second possibility is that $p_i$ sends the nullify$(v)$ message at $t$ because the conditions of lines \ref{n3a}-\ref{n3b} hold for $p_i$ at $t$.  The fact that some correct processor adds  $b$ into their local value $\mathtt{blocks}$ means that Decode does not output $\bot$ when given any $2f+1$ certified fragments of $b$ as input.   By Lemma \ref{eX1}, no other block for view $v$ receives an M-certificate. This gives the required contradiction. 
\end{proof}

\begin{lemma}[Consistency] The protocol satisfies Consistency.
\end{lemma}
\begin{proof}
Towards a contradiction, suppose some block $b_1$ with $b_1.\text{view}=v_1$ receives an L-certificate and is finalised by some correct processor, and that for some  \emph{least}  $v_2\geq v_1$ some block $b_2$ satisfies:
\begin{enumerate}
\item $b_2.\text{view}=v_2$;
\item   $b_1$ is not an ancestor of $b_2$, and;
\item $b_2$  receives an M-certificate. 
\end{enumerate}
 From Lemma \ref{eX1},  it follows that $v_2>v_1$. According to clause  (i) from the definition of a votable fragment,  correct processors will not vote for $b_2$ in line \ref{evote1} until they receive an M-certificate for its parent, $b_0$ say.  Correct processors will not vote for $b_2$ in line \ref{evote2} until $b_2$ has already received an M-certificate, meaning that at least $f+1$ correct processors must first vote for $b_2$ via line \ref{evote1}, and $b_0$ must receive an M-certificate. By our choice of $v_2$, it follows that  $b_0.\text{view}<v_1$. This gives a contradiction, because, by clause (ii) from the definition of a votable fragment for view $v_2$,  correct processors would not vote for $b_2$  in line \ref{evote1} without receiving an N-certificate for view $v_1$.  By Lemma \ref{eX2}, such an N-certificate cannot exist. So,  block $b_2$  cannot receive an M-notarisation or  M-certificate  (and no correct processor votes for $b_2$ via either line \ref{evote1} or \ref{evote2}).
\end{proof}

\subsection{Liveness} \label{elivesec}

We let $\mathtt{blocks}_i$ be the variable $\mathtt{blocks}$ as locally defined for $p_i$, and let $\mathtt{blocks}_i(t)$ be $\mathtt{blocks}_i$ as defined at the end of timeslot $t$. 

\begin{lemma}[Agreement on $\mathtt{blocks}$] \label{blockslem} 
If $p_i$ and $p_j$ are correct and $t$ is the least timeslot such that $b\in \mathtt{blocks}_i(t)$, then for $t':=\text{max} \{ t, \text{GST} \}+ 2\Delta +s$, $b\in  \mathtt{blocks}_j(t')$. 
\end{lemma} 
\begin{proof} 
The proof is by induction on $b.\text{view}$. The result holds trivially for $b_{\text{gen}}$. So, suppose the conditions in the statement of the lemma hold, $v:=b.\text{view}>0$,  and that the claim holds for all previous views (and so, for the parent of $b$). 

Since $b\in \mathtt{blocks}_i(t)$, $p_i$ receives an M-certificate for $b$ by $t$, and so disseminates such an M-certificate by $t$ (line \ref{Mdis}). It follows that all correct processors receive an M-certificate for $b$ by $\text{max} \{ t, \text{GST} \}+ \Delta$.  At $t+s$, $p_i$ executes lines \ref{f1a}-\ref{f1b} with respect to $b$. This ensures that at least $2f+1$ correct processors receive their own certified fragment of $b$ (as well as an M-certificate for $b$) by  $\text{max} \{ t+s, \text{GST} \}+ \Delta$. Lines \ref{f2a}-\ref{f2b} ensure that those correct processors will disseminate their fragments by the latter timeslot. All correct processors will therefore receive an M-certificate and fragments sufficient to decode the payload of $b$ by $\text{max} \{ t+s, \text{GST} \}+ 2\Delta$. By the induction hypothesis, and since the parent of $b$ is in $\mathtt{blocks}_i(t)$, it follows that for $t':=\text{max} \{ t, \text{GST} \}+ 2\Delta +s$ and for any correct processor $p_j$, $b\in  \mathtt{blocks}_j(t')$ as required. 
\end{proof} 

\begin{lemma}[Timely view entry] \label{timev} 
If correct $p_i$ enters view $v$ at $t$, then all correct processors enter view $v$ by $t':=\text{max} \{ t, \text{GST} \}+ 2\Delta +s$. 
\end{lemma} 
\begin{proof} 
The proof is by induction on $v$. For $v=1$, the claim is immediate. Suppose $v>1$ and that the claim holds for all previous views.  Let $p_i$, $t$ and $t'$ be as in the statement of the lemma. By the induction hypothesis, it follows that all correct processors enter view $v-1$ by $t'$. Processor $p_i$ either enters view $v$ at $t$ upon finding a view $v-1$ block in $\mathtt{blocks}$, or upon receiving an N-certificate for view $v-1$. In the former case, the claim follows by Lemma \ref{blockslem}. In the latter case, the claim follows because $p_i$ disseminates the N-certificate at $t$ (line \ref{Ndis}). 
\end{proof}

\begin{lemma}[Progression through views] Every correct processor enters every view $v\in \mathbb{N}_{\geq 1}$.
\end{lemma}
\begin{proof}  Towards a contradiction, suppose that some correct processor $p_i$ enters view $v$, but never enters view $v+1$.  From Lemma \ref{timev}, it follows that:
\begin{itemize}
\item All correct processors enter view $v$;
\item No correct processor leaves view $v$.
\end{itemize}
Each correct processor $p_j$ disseminates at least one of: 
\begin{itemize} 
\item[(a)] A nullify$(v)$ message; 
\item[(b)]  A vote for some view $v$ block $b$, together with the certified fragment of $b$ at $j$.  
\end{itemize} 
If there exists a view $v$ block $b$ such that at least $2f+1$ correct processors vote for $b$ and disseminate their corresponding certified fragments of $b$, and if Decode does not output $\bot$ when given those fragments as input, then we reach an immediate contradiction.  If there exists $b$ such that at least $2f+1$ correct processors vote for $b$ and disseminate their corresponding certified fragments of $b$, and if Decode outputs $\bot$ when given those fragments as input, then we also reach a contradiction: in this case, all correct processors will send nullify$(v)$ messages  (line \ref{esendN2}). So, suppose the previous two cases do not apply. If $p_j$ is a  correct processor that votes for a view $v$ block $b$, it follows that $p_j$  receives messages from at least $(n-f)-(2f)=n-3f\geq 2f+1$ processors, each of which is either:
\begin{enumerate}
\item[(i)]  A nullify$(v)$ message, or;
\item[(ii)]  A vote for a view $v$ block different than $b$.
\end{enumerate}
 So, the conditions of lines \ref{ebeginN}-\ref{eendN} are eventually satisfied, meaning that $p_j$ sends a nullify$(v)$ message (line \ref{esendN}). Any correct processor that does not vote for a view $v$ block also sends a nullify$(v)$ message, so all correct processors send nullify$(v)$ messages. All correct processors therefore receive a nullification for view $v$, giving the required contradiction.
\end{proof}

\begin{lemma}[Correct leaders finalise blocks] \label{eL1} If $p_i=\mathtt{lead}(v)$ is correct, and if the first correct processor to enter view $v$ does so after GST, then $p_i$ disseminates a block that is finalised by all correct processors.
\end{lemma}
\begin{proof} Suppose $p_i=\mathtt{lead}(v)$ is correct and that the first correct processor  to enter view $v$ does so at timeslot $t\geq\text{GST}$. From Lemma \ref{timev}, it follows that all correct processors (including $p_i$) enter view $v$ by $t+2\Delta+s$. 
Processor $p_i$ therefore disseminates a new block $b$  by $t+2\Delta+s$, which is received by all processors by $t+3\Delta+s$.  Let $b'$ be the parent of $b$ and set $v':=b'.\text{view}$. Since $b'\in \mathtt{blocks}_i(t+2\Delta+s)$, it follows from Lemma \ref{blockslem} that, for every correct $p_j$, $b'\in \mathtt{blocks}_j(t+4\Delta+2s)$.  Since $p_i$ has entered view $v$ at $t+2\Delta+s$, it must also have received N-certificates for all views in the open interval $(v',v)$ by this time, and all correct processors  receive these by $t+3\Delta +s$. All correct processors therefore vote for $b$ (by either line \ref{evote1} or \ref{evote2}) by $t+4\Delta+2s$,  before any correct processor sends a nullify$(v)$ message. The block $b$ therefore receives an L-notarisation and is finalised by all correct processors, as claimed.
\end{proof}

\begin{lemma}[Liveness] \label{eL2}The protocol satisfies Liveness.
\end{lemma}
\begin{proof}
Suppose correct $p_i$ receives the transaction $\text{tr}$. Let $v$ be a view with $\mathtt{lead}(v)=p_i$ and such that the first correct processor to enter view $v$ does so after GST.
By Lemma \ref{eL1}, $p_i$ will disseminate a block that is finalised by all correct processors.  From the definition of the ProposeBlock procedure, it follows that $\text{tr}$ will be included in the payload of an ancestor of $b$.  It follows that $\text{tr}$ is finalised for all correct processors upon finalisation of $b$. 
\end{proof}

\subsection{An optimisation reducing timeouts to $3\Delta$}
 According to Lemma \ref{eL1}, every view with a correct leader that begins after GST will produce a new finalised block if we set timeouts to $4\Delta+2s$. However, we can reduce timeouts with a small modification to the pseudocode, if we are prepared to accept the trade-off that view $v$ is then only guaranteed to produce a new finalised block if view $v-1$ begins after GST and the leaders of both views $v-1$ and $v$ are correct. 

\vspace{0.2cm} 
To do so, note first that while processors cannot immediately vote for a view $v$ block $b$ until they have received sufficiently many certified fragments of the parent $b'$ to recover its payload, we \emph{can} allow a correct processor to enumerate the parent $b'$ into their local value $\mathtt{blocks}$ upon receiving an M-certificate for the child $b$. This is because the existence of an M-certificate for $b$ suffices to ensure $f+1$ correct processors have already enumerated $b'$ into their local value $\mathtt{blocks}$. 

\vspace{0.2cm} 
We can therefore modify the way in which a correct processor $p_i$ defines its local value $\mathtt{blocks}$ as follows.  At any timeslot, $\mathtt{blocks}$ contains $b_{\text{gen}}$ and all blocks $b$ such that (i)-(iii) below are all satisfied: 
\begin{itemize}
\item[(i)] $\mathtt{certificates}$ contains an  M-certificate for $b$.
\item[(ii)] Either $\mathtt{certificates}$ contains an M-certificate for a child of $b$ or: 
\begin{enumerate} 
\item[(a)] There exists $I\subset [n]$ with $|I|=2f+1$ such that, for each $j\in I$, $p_i$ has received $(b,j,c_j,\pi_j)$ from $p_j$, which is a certified fragment of $b$ at $j$, and;
\item[(b)]  On input $(b.\text{tag}, \{ (c_j,\pi_j) \}_{j\in I})$, Decode does not output $\bot$. 
\end{enumerate} 
\item[(iii)] There exists $b'\in \mathtt{blocks}$ with $H(b')=b.\text{par}$, i.e., the parent of $b$ is in $\mathtt{blocks}$. 
\end{itemize} 

Now suppose we set timeouts to $3\Delta$. In what follows, we suppose that $2s\leq \Delta$, that view $v-1$ begins after GST (i.e., the first correct processor to enter the view does so after GST) and that the leaders of views $v-1$ and $v$ are both correct. There are four cases to consider. 

\vspace{0.2cm} 
\noindent \textbf{Case 1}. The first correct processor to enter view $v$ does so upon receiving an N-certificate for view $v-1$. The leader of view $v$ proposes a block $b$ with parent $b'$ such that $b'.\text{view}<v-1$. In this case, suppose the first correct processor to enter view $v-1$ does so at $t_0$, while the first correct processor to enter view $v$ does so at $t_1$. Note that $t_1\geq t_0+3\Delta$.  By Lemma \ref{timev}, $\mathtt{lead}(v)$ enters view $v-1$ by $t_0+2\Delta +s$ and $b'$ is in its local value $\mathtt{blocks}$ by this time. By Lemma \ref{blockslem}, and since $2s\leq \Delta$, $b'$ is in $\mathtt{blocks}_j$ for every correct processor $p_j$ by $t_1+2\Delta$. Processor $\mathtt{lead}(v)$ enters view $v$ by $t_1+\Delta$ and runs ProposeBlock by this time. All correct processors vote for $b$ by $t+2\Delta$. 

\vspace{0.2cm} 
\noindent \textbf{Case 2}. The first correct  processor to enter view $v$ does so upon receiving an N-certificate for view $v-1$. The leader of view $v$ proposes a block $b$ with parent $b'$ such that $b'.\text{view}=v-1$. In this case, suppose the first correct processor to enter view $v$ does so at $t$. Then all correct processors, including $\mathtt{lead}(v)$, do so by $t+\Delta$. All correct processors receive an M-certificate for $b'$ by $t+2\Delta$, and also add the parent $b''$ of $b'$ and all ancestors of $b''$ into their local value $\mathtt{blocks}$ by this time. Since $\mathtt{lead}(v-1)$ is correct, all correct processors  receive sufficiently many certified fragments of $b'$ to recover the payload by $t+3\Delta$, and vote for $b$ by this time. 

\vspace{0.2cm} 
\noindent \textbf{Case 3}. The first correct processor to enter view $v$ does so upon receiving an M-certificate for a view $v-1$ block $b'$. The leader of view $v$ proposes a block $b$ with $b'$ as parent. In this case, suppose  the first correct processor to enter view $v$ does so at $t$. Then all correct processors receive an M-certificate for $b'$ by $t+\Delta$. Since $\mathtt{lead}(v-1)$ is correct, all correct processors receive certified fragments sufficient to recover the payload of $b'$ by $t+2\Delta$, and have entered view $v$ and added $b'$ and all ancestors to their local value $\mathtt{blocks}$ by this time. In particular, $\mathtt{lead}(v)$ enters view $v$ by $t+2\Delta$, and all correct processors vote for their proposal $b$ by $t+3\Delta$. 

\vspace{0.2cm} 
\noindent \textbf{Case 4}. The first correct  processor to enter view $v$ does so upon receiving an M-certificate for a view $v-1$ block $b''$. The leader of view $v$ proposes a block $b$ with $b'$ as parent such that  $b'.\text{view}<v-1$.  In this case, suppose the first correct  processor to enter view $v-1$ does so at $t_0$, while the first correct processor $p_i$ to enter view $v$ does so at $t_1$. Then all correct processors receive an M-certificate for $b''$ by $t_1+\Delta$ and add all ancestors of the parent of $b''$ to their local value $\mathtt{blocks}$ by this time. Since some correct processors receive their certified fragments of $b''$ from $\mathtt{lead}(v-1)$ by $t_1$ and $\mathtt{lead}(v-1)$ is correct, all correct processors receive their certified fragments of $b''$ by $t_1+\Delta$, and so would vote for $b''$ by this time if they have not already timed out. Since $\mathtt{lead}(v)$ leaves view $v-1$ upon receiving an N-certificate for view $v-1$, it follows that $t_1\geq t_0+2\Delta$. By Lemma \ref{timev}, $\mathtt{lead}(v)$ enters view $v-1$ by $t_0+2\Delta +s$ and, by Lemma \ref{blockslem} and the fact that $2s\leq \Delta$, all correct processors add $b'$ to their local value $\mathtt{blocks}$ by $t_0+5\Delta\leq t_1+3\Delta$.  As in Case 3, all correct processors enter view $v$ by $t_1+2\Delta$. All correct processors therefore receive their certified fragment of $b$ from $\mathtt{lead}(v)$ by $t_1+3\Delta$ and it is votable for view $v$ by that time.

\end{document}